\documentclass[11pt,a4paper]{amsart}

\usepackage[utf8]{inputenc}
\usepackage{bm}
\usepackage{enumitem}
\usepackage{graphicx}
\usepackage{hyperref}
\usepackage{newtxtext,newtxmath}
\usepackage{tikz}
\usepackage{amsmath}
\usepackage{comment}
\usepackage{pdflscape,adjustbox}
\usepackage{blkarray}

\usepackage[linesnumbered,ruled,vlined]{algorithm2e}
\SetKwInput{KwInput}{Input}                
\SetKwInput{KwOutput}{Output}              

\theoremstyle{plain}
\newtheorem{proposition}{Proposition}
\newtheorem{theorem}{Theorem}

\theoremstyle{definition}
\newtheorem{definition}{Definition}
\theoremstyle{remark}
\newtheorem{remark}{Remark}
\newtheorem{example}{Example}

\DeclareMathOperator{\Ker}{ker}

\DeclareMathOperator{\Expectation}{\mathbb E}

\DeclareMathOperator{\Supp}{Supp}

\newcommand{\aval}[1]{\left\vert#1\right\vert}

\newcommand{\expectof}[1]{\Expectation\left(#1\right)}

\newcommand{\integers}{\mathbb Z}
\newcommand{\jointof}[2]{\mathcal P \left(#1,#2\right)}

\newcommand{\moves}{\mathcal M}

\newcommand{\one}{\bm 1}

\newcommand{\reals}{\mathbb{R}}

\newcommand{\scalarof}[2]{\left\langle#1,#2\right\rangle}
\newcommand{\setof}[2]{\left\{#1 \,\middle|\, #2 \right\}}
\newcommand{\set}[1]{\left\{#1\right\}}
\newcommand{\suppof}[1]{\Supp\left(#1\right)}

\usepackage{cleveref}

\begin{document}

\title{Finite space Kantorovich problem with an MCMC of table moves}

\author{Giovanni Pistone}
\address{Collegio Carlo Alberto}
\email{giovanni.pistone@carloalberto.org}
\author{Fabio Rapallo}
\address{Dipartimento DIEC, Universit\`a di Genova}
\email{fabio.rapallo@unige.it}
\author{Maria Piera Rogantin}
\address{Dipartimento di Matematica, Universit\`a di Genova}
\email{rogantin@dima.unige.it}


\begin{abstract}
In Optimal Transport (OT) on a finite metric space, one defines a distance on the probability simplex that extends the distance on the ground space. The distance is the value of a Linear Programming (LP) problem on the set of non-negative-valued 2-way tables with assigned probability functions as margins. We apply to this case the methodology of moves from Algebraic Statistics (AS) and use it to derive a Monte Carlo Markov Chain (MCMC) solution algorithm.
\end{abstract}


\maketitle


\section{Introduction}

In the present paper, we aim to show a connection between Optimal Transport (OT) and Algebraic Statistics (AS).

Modern OT was started by Kantorovich in 1939 and a new wave of development was initiated by Villani  \cite{villani:2003-topics}. In the present paper we use also an earlier result obtained by Gini \cite{gini:1914dissomiglianza}. A (finite) sample space $X$ and a cost function $c  \colon  X \times X \rightarrow {\mathbb R}$ are given. The set of joint probability functions $\gamma$ on $X \times X$ with given margins $\mu$ and $\nu$ is called the set of couplings, $\gamma \in \jointof \mu \nu$. In OT, one looks for an element that minimizes the expected value $c(\gamma) = \sum_{x,y \in X} c(x,y) \gamma(x,y)$. There is a rich general theory, see, for example, the textbook by Santambrogio \cite{santambrogio:2015OTAP}, but here we restrict our attention to the finite state space case.

AS was started by the paper Diaconis and Sturmfels \cite{diaconis|sturmfels:98} and by the book Pistone, Riccomagno, and Wynn \cite{pistone|riccomagno|wynn:2001}. In particular, the first paper deals with an algebraic method for constructing an irreducible random walk on the space of multi-way contingency tables with given margins. Each step of the random walk is associated with a move, that is, a table with zero margins, that subtracted to an initial table, produces a new table with the same margins. Basic results on contingency tables are to be found in Fienberg \cite{fienberg:80}.

We extend this idea to general tables, that is, tables not restricted to be integer-valued, and apply it to OT on a finite state space. To this aim, we provide a detailed study of the geometry of moves with continuous values. This paper considers both topics in computational algebra and in computational statistics. As an application, we define an MCMC algorithm for the computation of the optimal value and the optimal coupling in the case of a discrete sample space. Many special algorithms have been developed, see a general overview in Peyr\'e and Cuturi \cite{peyre|cuturi:2019}. Our algorithm is intended to be an alternative proposal.

The paper is organised as follows. In \Cref{sec:distance} we review the generalities and discuss the algebra of moves, considering both the linear algebra and the group algebra of moves. The Kantorovich problem is a special Linear Programming (LP) problem that we outline both as a primal and as a dual problem. In \Cref{sec:gini} we prove that a class of basic moves connects all couplings. The results are generalized to the tri-variate case in \Cref{sec:three}. Based on that theory, in \Cref{sec:algo} we provide
a MCMC algorithm to compute solutions of the minimal cost problem.

\section{Tables, cost, moves}\label{sec:distance}

Let $X$ be a set with $n$ points and let $\Delta(X)$ be its probability simplex.

Given probability functions $\mu,\nu \in \Delta(X)$, the joint probability function $\gamma \in \Delta(X \times X)$ is a \emph{coupling} (also called \emph{transport plan}) of $(\mu,\nu)$, if $\mu$ and $\nu$ are the two margins of $\gamma$. The set of all couplings $\mathcal P(\mu,\nu)$ is the polyhedron defined by the intersection of $\Delta(X \times X)$ with the $2n$ affine hyperplanes
\begin{equation}\label{eq:margins}
 \sum_{y\in X} \gamma(x,y) = \mu(x) \ , \quad  \sum_{x\in X} \gamma(x,y) = \nu(y) \ , \qquad x,y \in X \ .
\end{equation}
The number of independent constraints is $2n -1$ and the dimension of the polyhedron is $(n-1)^2$. This polyhedron is bounded, then it is a polytope. See the relevant convexity theory in \cite[Ch.~I-II]{barvinok:2002}.

As we are dealing with functions defined on points in a product space, $\gamma(x,y) \in \reals_+$, $(x,y) \in X \times X$, we consider the following definition. See the relevant graph theory in \cite{bollobas:1998}.

\begin{definition}
The \emph{support} of the  coupling $\gamma$ is
\begin{equation*}
    \suppof \gamma = \setof{(x,y)}{\gamma(x,y) > 0} \ .
\end{equation*}
It is identified with a directed, possibly non-simple, graph with vertex set $X$ and edge set $\suppof \gamma$. By abuse of language, the graph itself is the support of $\gamma$.
\end{definition}

If we add weights $\gamma$ to the graph $\suppof \gamma$, we obtain a weighted graph. Vertices of the coupling polytope are characteristic in that they have a small support.

\begin{proposition}\label{prop:couplings-polytope}
If $\widetilde \gamma$ is a vertex of the coupling polytope $\mathcal P(\mu,\nu)$, then its support $\suppof {\widetilde \gamma}$ has at most $(2n-1)$ edges.
\end{proposition}

\begin{proof}
This theorem is due to Brualdi \cite{brualdi:2006}. See a proof based on the representation of the support as a bipartite graph in \cite[\S 3.4]{peyre|cuturi:2019}.
\end{proof}

As $2n-1 = n + (n-1)$, the condition in the proposition above could be realized by a graph that has $n$ loops $x \to x$, $x \in X$, and other edges to form a tree. This is not always the case, as the \Cref{ex:2x2} below shows.

Notice that, for a vertex $\widetilde \gamma$, the $2n$ marginalization equations  in \Cref{eq:margins} have $2n-1$ non-zero unknowns $\gamma(x,y)$, $(x,y) \in \suppof {\widetilde \gamma}$, so that an extremal coupling is uniquely determined by its support.

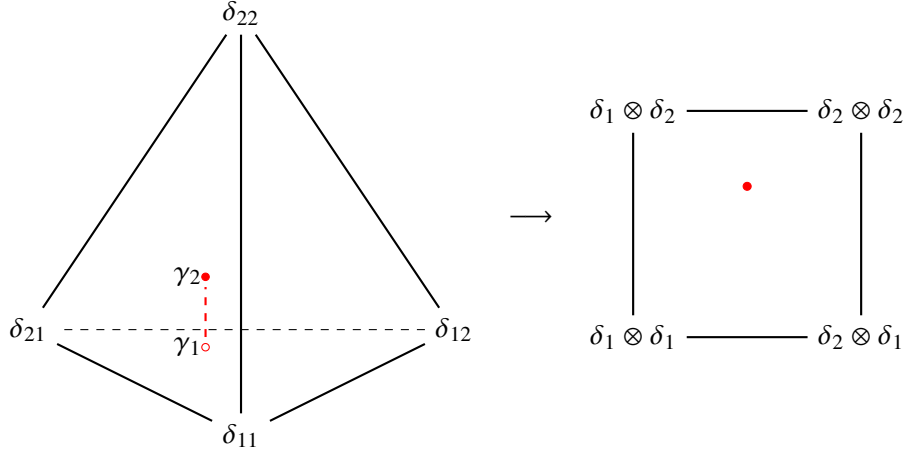
\begin{figure}
  \begin{tabular}{ccc}
    \begin{tikzpicture}[scale=2.8,baseline = (current bounding box.center)]
\node (n00) at (1,-.5) {$\delta_{11}$};
\node (n10) at (0,0) {$\delta_{21}$};
\node (n01) at (2,0) {$\delta_{12}$};
\node (n11) at (1,1.5) {$\delta_{22}$};
\node (K1) at (5/6,-1/12) {};
\node (gK1) at (5/6-1/12,-1/12) {$\gamma_1$};
\node (K2) at (5/6,1/4) {};
\node (gK2) at (5/6-1/12,1/4) {$\gamma_2$};
\draw[dashed] (n01) -- (n10);
\filldraw [red] (5/6,1/4) circle (.5pt);
\draw [red] (5/6,-1/12) circle (.5pt);
\foreach \from/\to in {n00/n10,n00/n01,n11/n10,n11/n01,n00/n11}
\draw[thick] (\from) -- (\to);
\draw [thick,red,dashed] (K1) -- (K2); 
\end{tikzpicture} & $\longrightarrow$ &
    \begin{tikzpicture}[scale=3,baseline = (current bounding box.center)]
\node (n00) at (0,0) {$\delta_1\otimes\delta_1$};
\node (n10) at (1,0) {$\delta_2\otimes\delta_1$};
\node (n01) at (0,1) {$\delta_1\otimes\delta_2$};
\node (n11) at (1,1) {$\delta_2\otimes\delta_2$};
\foreach \from/\to in {n00/n01,n01/n11,n11/n10,n10/n00}
\draw[thick] (\from) -- (\to);
\filldraw [red] (1/2,2/3) circle (.5pt);
\end{tikzpicture}
  \end{tabular}
  \caption{See \cref{ex:2x2}. The arrow is the marginalization function of the probability simplex $\Delta(\set{1,2}^2)$ to the product of the two marginal simplexes $\Delta(\set{1,2}) \times \Delta(\set{1,2})$. Each vertex of the left simplex is mapped to a vertex of the right polytope, $\delta_{ij} \mapsto \delta_i \otimes \delta_j$. The dashed segment from $\gamma_1$ to $\gamma_2$ represents the coupling polytope of the margins represented by the circle in the right polytope. Notice that $\gamma_1$ belongs to the facet opposite to $\delta_{22}$, while $\gamma_2$ belongs to the facet opposite to $\delta_{12}$.}\label{fig:2x2}
  \end{figure}

\begin{example}\label{ex:2x2}
Let us consider $X = \set{1,2}$. The probability simplex $\Delta(X \times X)$ is the 3-simplex of Figure \ref{fig:2x2}. The dashed segment represents the set of couplings $\mathcal P((1/2,1/2),(2/3,1/3))$. The two end-points are
\begin{equation*}
\gamma_1 =
    \begin{pmatrix}
      1/6 & 1/3 \\ 1/2 & 0
    \end{pmatrix} \ , \quad \gamma_2 =
\begin{pmatrix}
  1/2 & 0 \\ 1/6 & 1/3
\end{pmatrix} \ .
\end{equation*}
The supports of $\gamma_1$ and $\gamma_2$ have $2\cdot2-1 = 3$ arcs. The support of $\gamma_2$ is a looped tree, while the support of $\gamma_1$ is not because of the cycle  $1 \rightleftarrows 2$.
The support of each non-vertex coupling $\gamma = (1-\lambda)\gamma_1 + \lambda \gamma_2$, $0 < \lambda < 1$, has $4$ arcs.
\end{example}

\medskip

The notion of couplings has a related setup in the context of the study of integer-valued tables with given margins. Given a table $T = [n(i,j)]_{i,j=1}^n \in \integers_+^{n \times n}$, the grand total is $n(+,+) = \sum_{i,j=1}^n n(i,j)$ and the margins are $n(\cdot,+)$, $n(+,\cdot)$. The corresponding probability function is defined by $\gamma(i,j) = n(i,j)/N$, with $i,j \in \{1, \ldots, n\}$. Conversely, if $\gamma \in \Delta(X \times X)$ has rational values, it comes from a table. See the extensive treatments in \cite{fienberg:80} and \cite{sullivan:2018book}.

Let $c \colon X \times X \to \reals_+$ be a non-negative valued function to be interpreted as the cost. The \emph{cost} of a coupling $\gamma$ (c-cost) is
\begin{equation}\label{eq:c-cost}
c(\gamma) = \sum_{x,y \in X} c(x,y)\gamma(x,y) \, .
\end{equation}
We are interested in minimizing the expected cost over the polytope of couplings. The \emph{Kantorovich cost} (K-cost) is
\begin{equation}\label{eq:K-cost}
K_c(\mu,\nu) = \inf \setof{c(\gamma)}{\gamma \in \jointof \mu \nu} \ .
\end{equation}

Especially, when the cost is a distance $d$, the minimum cost defines a distance on the simplex $\Delta(X)$, the \emph{Kantorovich distance} (K-distance), namely,
\begin{equation}\label{eq:K-distance}
  d(\mu,\nu) = \inf \setof{\sum_{x,y \in X} d(x,y) \gamma(x,y)}{\gamma \in \mathcal P(\mu,\nu)} \ .
\end{equation}
The distance case is considered in detail in \cite{montrucchio|pistone:2019-arXiv:1905.07547v5}.

As the simplex is a compact set, the optimal value is always obtained at some optimal coupling.

In the case of equality of the two margins $\mu = \nu$, the distance is zero because there is a coupling whose support consists of loops only, where $d(x,x)=0$. When the coupling is defined by the independence, $\gamma = \mu \otimes \mu$, the Kantorovich value is a Gini index of dispersion of $\mu$, see the monograph by Yitzhaki and Schechtman \cite{yitzhaki|schechtman:2013}.

The Kantorovich problem defined above is a special LP problem, in that we want to find the minimum of a linear function subject to equality and inequality constraints. It follows immediately from the definition that there exists a face of $\mathcal P(\mu,\nu)$ whose elements $\widetilde \gamma$ are optimal, that is, $c(\widetilde \gamma)=K_c(\mu,\nu)$ or, in the distance case,
$d(\mu,\nu) = \sum_{x,y} d(x,y) \widetilde \gamma(x,y)$. Generically, the set of solutions will be a vertex of the coupling polytope, hence subject to the support constraints of \Cref{prop:couplings-polytope}.

Let us discuss an equivalent form of the Kantorovich problem.

The marginalization operator is
\begin{equation*}
\Pi \colon   \reals^{X \times X} \ni f \mapsto \left(\sum_y
  f(\cdot,y),\sum_x f(x,\cdot)\right) \in \reals^X\oplus \reals^X \ ,
\end{equation*}
and $\Ker \Pi$ is the set of all functions  $f \colon X \times X \to
\reals$ whose margins are zero. It follows that
\begin{equation*}
  \mathcal P(\mu,\nu) = \setof{\mu \otimes \nu - f}{f \in \Ker \Pi,
    \mu \otimes \nu \geq f} \ ,
\end{equation*}
so that
\begin{equation*}
   K_c(\mu,\nu) = \sum_{x,y} c(x,y)\mu(x)\nu(y)  - \sup\setof{\sum_{x,y}
     c(x,y)f(x,y)}{f \in \Ker \Pi, \mu \otimes \nu \geq f} .
 \end{equation*}

 Let us show that the convex set
 \begin{equation*}\mathcal A = \setof{f}{f \in \Ker \Pi, \mu \otimes
     \nu \geq f}
 \end{equation*}
 is, in fact, a compact convex set. In fact, for each $f \in \mathcal
 A$ and all $(x,y)$, it holds
 \begin{equation*}
   f(x,y) = - \sum_{u \neq y} f(x,u) \geq - \sum_{u \neq y}
   \mu(x)\nu(u) = \mu(x)\nu(y) - \mu(x) \geq - \mu(x) \ .
 \end{equation*}
The same argument applies to the other variable, so that $f(x,y) \geq -
(\mu(x) \wedge \nu(y))$. In conclusion,
\begin{equation*}
\mathcal A = \Ker \Pi \cap \setof{f}{\mu(x)\nu(y) \geq f(x,y)
  \geq - \mu(x)\wedge\nu(y), x,y \in X} \ .
\end{equation*}

In turn, this allows to give a proof of the following continuity result.

\begin{proposition}
\label{prop:continuity}
The mapping $(\mu,\nu) \mapsto K_c(\mu,\nu)$ is continuous in the topology of $\reals^X \oplus \reals^X$.
\end{proposition}

\begin{proof}
This is an application of Berge's Maximum Theorem, see, for example, \cite[\S~17.5]{aliprantis|border:2006}. Here is a sketch of a   proof. As the function to optimize is continuous, one has to show that the mapping $(\mu,\nu) \mapsto \mathcal A(\mu,\nu)$ is both upper and lower hemicontinous, see the definitions in  \cite[\S~17.2]{aliprantis|border:2006}. In our case, upper
hemicontinuity follows from the compactness. Lower hemicontinuity is proved by considering a sequence $(\mu_n,\nu_n)$ converging to $(\mu,\nu)$ and noting that the elements of the sequence $\mathcal A(\mu_n,\nu_n)$ are convex and contained in an $\epsilon$-neighborhood of $\mathcal A(\mu,\nu)$.
\end{proof}

As the Kantorovich problem is an LP problem, the duality theory applies, see, for example,  \cite[\S~IV.8]{barvinok:2002}. \Cref{eq:c-cost,eq:K-cost} can be  written in primal standard form as
\begin{equation*}
  K_c(\mu,\nu) = \inf_\gamma \scalarof c \gamma \qquad \text{subject to} \qquad \Pi \gamma = (\mu,\nu) \ , \quad \gamma \geq 0 \ .
\end{equation*}

The equivalent dual standard form is
\begin{equation*}
 \sup_{(\phi,\psi)} \scalarof {(\mu,\nu)}{(\phi,\psi)} \qquad \text{subject to}
 \qquad \Pi^t (\phi,\psi) \leq c \ ,
\end{equation*}
that is,
\begin{equation}\label{eq:dual-LP}
  K_c(\mu,\nu) = \sup \setof {\sum_{z\in X} \phi(z) \mu(z) +
    \sum_{z\in X} \psi(z) \nu(z)}{\phi \oplus \psi \leq c} \ ,
\end{equation}
In fact, $\Pi^t(\phi_1,\phi_2) = \phi_1 \oplus \phi_2$ in the functional representation and  $= \phi_1 \one^t + \one \phi_2^t$ in the matrix representation.

In this paper, we restrict our attention to the primal problem. However, the dual problem is interesting in that the domain does not depend on $\mu$, $\nu$, but it depends on the cost $c$ only.

Let us observe that the feasibility domain $\set{\phi \oplus \psi}$ in the dual problem can be further restricted. For a full presentation of the following argument, see \cite[\S~1.6]{santambrogio:2015OTAP}. If
$\phi(x) + \psi(y) \leq c(x,y)$, then
$\phi_1(x) = \inf_y c(x,y) - \psi(y)$ has the following properties:
\begin{enumerate}[label=(\alph*)]
\item $\phi(x) \leq \phi_1(x)$;
\item $\phi_1(x) + \psi(y) \leq c(x,y)$;
\item For each distance $d$ on $X$, there is a constant $K$ depending of $d$ and $c$ only such that $\phi_1(z) - \phi_1(z') \leq K d(z,z')$.
\end{enumerate}

The same argument applies to $\psi$. In conclusion, the feasible domain can be restricted, without changing the maximum, to all pairs $(\phi,\psi)$ such that
\begin{equation}
  \label{eq:condition}
\phi(x) + \psi(y) \leq c(x,y) \ , \quad \phi(z) - \phi(z') \leq K d(z,z') \ , \quad \psi(z) - \psi(z') \leq K d(x,y) \ .
\end{equation}
In particular, the optimal pair satisfies all the conditions above.

When the cost $c$ is a distance (denoted, if any confusion could arise, by $d$), then the Kantorovich construction induces a distance on probability functions. Moreover, it is possible to define metric geodesics and hence, a proper geometry associated to the given distance. The following proposition provides the details. The   extension property is a key characteristic of the K-distance which is not shared by other statistical measures of divergence.

\begin{proposition}\label{prop:is-a-distance}
Assume that the cost function in \Cref{eq:K-distance} is a distance $d$.
  \begin{enumerate}
  \item The $K_d$ value is a distance that extends the ground distance, that is, the K-distance between two Dirac probability functions equals the distance between the respective supports.
      \item Given $\mu, \nu \in \Delta(X)$, the mixture curve $\mu(t) = (1-t)\mu + t \nu$, $0 \leq t \leq 1$, is a metric geodesic for the K-distance, that is,
        \begin{equation*}
         K_d(\mu(t),\mu(s)) = (t-s)
        K_d(\mu,\nu) \ , \quad  0 \leq s \leq t \leq 1 \ .
        \end{equation*}
        \item If $\widetilde \gamma$ is optimal for $d(\mu,\nu)$, then the coupling defined by
\begin{equation*}
    \widetilde \gamma(x,y;s,t) = \left[(1-t) \mu(x) + s \nu(y)\right] (x=y) + (t-s) \widetilde \gamma(x,y) \ ,
\end{equation*}
with $(x=y)=1$ if $x=y$, 0 otherwise, is optimal for $K_d(\mu(s),\mu(t))$.
\end{enumerate}
\end{proposition}
\begin{proof}
This proof is known from the quoted literature. We repeat it here for sake of completeness.

Given the existence of optimal couplings, we can write
    \begin{equation*}
      K_d(\mu,\xi) = \sum_{x,z} d(x,z)\gamma_1(x,z) \quad
      \text{and} \quad K_d(\xi,\nu) = \sum_{z,y} d(z,y)\gamma_2(z,y) \ .
    \end{equation*}
Moreover,
\begin{equation*}
  \gamma(x,y) = \sum_{\setof{z}{\xi(z) >  0}} \frac{\gamma_1(x,z)\gamma_2(z,y)}{\xi(z)}
\end{equation*}
defines a coupling $\gamma$ of $\mu$ and $\nu$ whose value is less than or equal to the sum of the two values. Notice that $d$ must be a distance because we want to use the triangle inequality to check the last statement.

The other two statements are proved together. First, one checks that $\widetilde \gamma(s,t)$ is indeed a coupling of $\mu(s) = (1-s)\mu + s \nu$ and $\mu(t) = (1-t)\mu + t \nu$, and its value is $(t-s)K_d(\mu,\nu)$. It follows that $K_d(\mu(0),\mu(s)) \leq s   K_d(\mu,\nu)$, $K_d(\mu(s),\mu(t)) \leq (t-s) K_d(\mu,\nu)$, and $K_d(\mu(t),\mu(1)) \leq (1-t) K_d(\mu,\nu)$. But none of the inequalities can be strict, because otherwise,
  \begin{multline*}
    K_d(\mu,\nu) \leq K_d(\mu(0),\mu(s)) + K_d(\mu(s),\mu(t)) +
    K_d(\mu(t),\mu(1)) < \\ (s + (t-s) + (1-t)) K_d(\mu,\nu) =
    K_d(\mu,\nu) \ .
  \end{multline*}
This concludes the proof.
\end{proof}

The previous proposition does not rule out the existence of multiple geodesics between two points.

We will take also advantage of the following definition from the algebraic theory of two-way contingency tables, see, for example, \cite{rapallo:2003-SJS} and \cite{aoki|hara|takemura:2012}. Remember that the affine space of the convex polytope $\mathcal P(\mu,\nu)$ is the vector space generated by the differences $\gamma_1-\gamma_2$, $\gamma_1,\gamma_2 \in \mathcal P(\mu,\nu)$. Clearly, the margins of the elements of the affine space are null.

\begin{definition} A \emph{move} is a real valued function $M$ defined on $X \times X$ and  with null margins, $\sum_x M(x,y) = \sum_y M(x,y) = 0$. An \emph{integer move} is an integer valued move. It is a \emph{simple move} if it  takes values in $\{-1,0,1\}$. It is a \emph{basic move} if it is of the form
  \begin{multline*}
\delta_{x_1}  \otimes \delta_{y_1} - \delta_{x_1}  \otimes \delta_{y_2} - \delta_{x_2}  \otimes \delta_{y_1} + \delta_{x_2}  \otimes \delta_{y_2} = \\
(\delta_{x_1}-\delta_{x_2})\otimes(\delta_{y_1}-\delta_{y_2}) \ , \quad x_1\neq x_2, y_1\neq y_2 \ .   \end{multline*}
\end{definition}

Throughout this paper, we write $\{M>0\}$ to denote the set of indices $\{(x,y) \ | \ M(x,y)>0\}$, and similarly for $\{M<0\}$.

Notice that there are $\binom n 2 ^2$ different basic moves up to the sign. They are not linearly independent. We prove below that, given a pivot point $(u,v)$, the $(n-1)^2$ basic moves of the type $(\delta_x - \delta_u) \otimes (\delta_y - \delta_v)$, with $x\ne u, y \ne v$, form a basis of the set of moves as vector space.

\begin{proposition}
The vector space $\moves(X \times X)$ of moves is the kernel of the marginalization mapping
\begin{equation*}
\Pi \colon  \reals^{n \times n} \ni A \mapsto (A\one,A^t\one) \in \left (\reals^n,\reals^n\right) \ .
\end{equation*}
The dimension of $\Ker \Pi$ is $(n-1)^2$. For each $u,v \in X$, the set of basic moves $(\delta_u - \delta_x)\otimes(\delta_v-\delta_y)$, $x,y \in X$, $x \neq u$ and $y \neq v$, is a basis of $\moves(X \times X)$. Moreover, it holds
\begin{equation}\label{eq:posbasisofM}
 M =\frac 1 { \# \set{M > 0}}\ \sum_{x,y} M(x,y) \sum_{u,v \colon M(u,v) > 0} (\delta_x-\delta_u)\otimes(\delta_y-\delta_v) \ .
\end{equation}
\end{proposition}

\begin{proof}
Note first that the image of the marginalization mapping is a space of dimension $(2n-1)$, precisely $\setof{(f,g) \in \reals^{2n}}{\sum_x f(x) = \sum_y g(y)}$. In fact $\one^t A \one = \one^t A^t \one$, and, given any pair of margins $f$ and $g$ such that $\sum_x f(x) = \sum_y g(y)$, the outer product $f \otimes g$ is a counter-image. It follows that the dimension of the kernel is $n^2 - (2n-1) = (n-1)^2$.

Every basic move $(\delta_{u}-\delta_{x})\otimes(\delta_{v} - \delta_{y})$ is clearly an element of the kernel. Let us find a basis of $\moves$. Let $M \in \moves$ and fix $u,v \in X$. As $M(u,v) = - \sum_{x \neq u} M(x,v) = \sum_{x \neq u,y \neq v} M(x,y)$, with straightforward computations one obtains
  \begin{equation*}
    M = \sum_{x \neq u, y \neq v} M(x,y) (\delta_x - \delta_u) \otimes (\delta_y - \delta_v) \, .
  \end{equation*}
\Cref{eq:posbasisofM} now follows immediately adding over all $u,v$ such that $M(u,v) > 0$.
\end{proof}

We have shown that every move $M$ is a linear combination of the $(n-1)^2$ basic moves $(\delta_x - \delta_u) \otimes (\delta_y - \delta_v)$, $x \neq u$ and $y \neq v$. In particular, all other basic moves are combination of these special moves. More generally, if $M$ is a simple move,
\begin{equation*}
    M = \sum_{M(x,y) = +1} (\delta_x - \delta_u) \otimes (\delta_y - \delta_v) - \sum_{M(x,y) = -1} (\delta_x - \delta_u) \otimes (\delta_y - \delta_v) \ .
\end{equation*}

In spite of the $(n-1)^2$ pivotal moves around $(u,v)$ form a linear basis of the vector space of moves, we will need to use all basic moves in order to perform a connected random walk that stays in the polytope $\jointof \mu \nu$, see \cite{sullivan:2018book}.

\begin{proposition}
  The move $M$ is the difference of two couplings, $\gamma, \overline\gamma\in\jointof \mu \nu$ if, and only if, both hold
  \begin{equation*}
    \sum_y \aval{M(x,y)} \leq 2 \mu(x) \quad \text{and} \quad \sum_x \aval{M(x,y)} \leq 2\nu(y) \ ,
  \end{equation*}
for all $x,y \in X$.
\end{proposition}

\begin{proof}
  If $\gamma, \overline \gamma \in \jointof \mu \nu$, then $M = \gamma - \overline \gamma$ is a move such that
  \begin{gather*}
    \sum_{y} \aval{M(x,y)} = \sum_y \aval{\gamma(x,y) - \overline \gamma(x,y)} \leq  \sum_y \gamma(x,y) + \sum_y \overline \gamma(x,y) = 2\mu(x)\ , \\
    \sum_{x} \aval{M(x,y)} = \sum_x \aval{\gamma(x,y) - \overline \gamma(x,y)} \leq  \sum_x \gamma(x,y) + \sum_x \overline \gamma(x,y) = 2\nu(x) \ .
  \end{gather*}
  Conversely, assume $M$ is a move, decomposed in its positive and negative part, $M=M^+-M^-$, such that
    \begin{gather*}
    \sum_{y} \aval{M(x,y)} = \sum_y (M^+(x,y) + M^-(x,y)) \leq  2\mu(x)\ , \\
    \sum_{x} \aval{M(x,y)} = \sum_x (M^+(x,y) + M^-(x,y)) \leq  2\nu(y)\ .
  \end{gather*}
  As $\sum_y M^+(x,y) = \sum_y M^-(x,y)$ and  $\sum_x M^+(x,y) = \sum_x M^-(x,y)$, we have
  \begin{gather*}
    a(x) =  \sum_y M^+(x,y) = \sum_y M^-(x,y) \leq \mu(x) \\
    b(y) = \sum_x M^+(x,y) = \sum_x M^-(x,y) \leq \nu(y) \ .
  \end{gather*}
Notice that $\sum_x a(x) = \sum_y b(y) = h$, so that there exist a non-negative $M^* \colon X \times X \rightarrow \mathbb R$ whose margins are $(\mu-a)$ and $(\nu - b)$, respectively, and whose grand total is $1-h$.

The equations
  \begin{equation*}
    \gamma(x,y) = M^+(x,y) + M^*(x,y) \ , \quad
    \overline\gamma(x,y) = M^-(x,y) + M^*(x,y) \ ,
  \end{equation*}
  provide the required coupling.
\end{proof}

\begin{proposition}
Every move $M$ is of the form
\begin{equation*}
  M = \alpha_1 F_1 + \cdots + \alpha_k F_k \ ,
\end{equation*}
where $\alpha_1, \dots, \alpha_k > 0$ and $F_1,\dots,F_k$ are simple
moves. Moreover, it is possible to choose the basic moves in such a
way that, for the sequence of remainders $M_j = M - (\alpha_1 F_1 +
\cdots + \alpha_j F_j)$, $j = 1,\dots,k$, it holds
\begin{equation*}
  \set{M_{j-1} > 0} \supset \set{M_{j} > 0} \quad \text{and} \quad
  \set{M_{j-1} < 0} \supset \set{M_{j} < 0} \, .
\end{equation*}
\end{proposition}

\begin{proof} Let $M$ be a move and define the two sets of indices $M_+ = \set{M > 0}$, $M_- = \set{M < 0}$. Without restriction of generality, assume that the first projection of $M_+$ has $n$ points. Let us define a directed bipartite graph with vertices $M_+ \cup M_-$ as follows. For each $(x,y) \in
M_+$ there is a edge going to $(x,\bar y)$ if $(x,\bar y) \in M_-$. For each $(x,y) \in M_-$ there is an edge going to $(\bar x,y)$ if $(\bar x,y) \in M_+$. Edges of the first type are horizontal in the table, while edges of the second type are vertical. At least one edge of the first type always exists for each $x$ because the sum over that row is null. The same holds for each column $y$.

By construction, there are at least $2n$ edges in the graph and at most $2n$ vertices. Hence, there is at least one irreducible cycle with even length, say $2m$. Fix a starting point in $M_+$ and enumerate the vertices as
\begin{multline*}
(x_1,y_1) \to (x_1, \bar y_1) \to (\bar x_1, \bar y_1) = (x_2, y_2)
\to \cdots \\ (\bar x_{m-1}, \bar y_{m-1}) = (x_m,y_m) \to (x_m, \bar y_m) \to (\bar x_m, \bar y_m) = (x_1, y_1)  \ .
\end{multline*}

Let us construct a simple move from the cycle above. Observe that
\begin{equation*}
F = \sum_{j=1}^ m \delta_{x_j} \otimes \delta_{y_j} - \sum_{j=1}^ m \delta_{x_j} \otimes
\delta_{\bar y_j} =  \sum_{j=1}^ m \delta_{\bar x_{j-1}} \otimes
\delta_{\bar y_{j-1}} - \sum_{j=1}^m \delta_{x_j} \otimes
\delta_{\bar y_j} \ ,
\end{equation*}
where the indices in the second expression are computed $\mod m$. The first expression shows that the first margin is zero, while the second expression shows that the second margin is zero.

For each positive $\alpha$, the move $M' = M - \alpha F$ subtracts from the values in $M_+$ and adds to the values in $M_-$. If $\alpha = \min \aval{M}$, then the operation cancels at least one non-zero value of $M$. As a consequence, $\# \suppof {M'} < \# \suppof{M}$.

Now the proposition is proved by a finite number of applications of the previous step.
\end{proof}

We are interested in the characterisation of moves which are the difference of two coupling, where the first one is fixed.

\begin{definition}
A move $M$ is \emph{admissible} for the coupling $\gamma \in \mathcal P(\mu,\nu)$ if $\overline \gamma = \gamma - \alpha M \geq 0$ for some $\alpha>0$, that is, $\overline \gamma = \gamma - \alpha M \in \jointof \mu \nu$. In other words, a move is admissible for $\gamma$, if, and only if, $\set{M > 0} \subset \suppof \gamma$.
\end{definition}

The couplings $\gamma$ and $\overline \gamma$ are related to each other through $M$ and $\alpha$. In particular, the cost of $\overline \gamma$ depends on $\alpha$, on the cost of $\gamma$, and on the cost of $M$. We are especially interested in $M$ being a simple move. In such
a case,
\begin{equation*}
  c(\overline \gamma) =  c(\gamma) - \alpha \left(\sum_{\{M= +1\}} c(x,y)
- \sum_{\{M = -1\}} c(x,y)\right) \ ,
\end{equation*}
so that the value $c(\overline \gamma)<c(\gamma)$ if, and only if,
\begin{equation*} 
  \sum_{\{M = +1\}} c(x,y) > \sum_{\{M = -1\}} c(x,y) \ .
\end{equation*}

Now, this property can be restated in a more specific form.

\begin{proposition}\label{prop:permutations} Let $M$ be a simple move and let $(x_i,y_i)$, $i =
  1,\dots,k$, be any sequence of $\set{M=+1}$. It holds
\begin{equation}\label{eq:move-permutation}
  M = \sum_{i=1}^k \delta_{x_i}\otimes\delta_{y_i} - \sum_{i=1}^k \delta_{x_i}\otimes\delta_{y_{\sigma(i)}} \ ,
\end{equation}
for a permutation $\sigma \in S_k$.
\end{proposition}
\begin{proof}
Clearly, the two sets $\set{M = +1}$ and $\set{M = -1}$ have the same number of points. Let $(\overline x_j,\overline y_j)$, $j=1,\dots,k$, be an arbitrary sequencing of the second one. The move is
\begin{equation*}
\sum_{i=1}^k \delta_{x_i}\otimes\delta_{y_i} - \sum_{j=1}^k \delta_{\overline x_j}\otimes\delta_{\overline y_{j}} \ .
\end{equation*}
The first margin is
\begin{equation*}
  \sum_y F(x,y) = \sum_{i=1}^k \delta_{x_i} - \sum_{j=1}^k \delta_{\overline x_j} = 0 \ .
\end{equation*}
It follows that $\overline x_j = x_{\sigma'(i)}$ for some permutation $\sigma'\in S_k$. Considering the second margin, we find $\overline y_j = y_{\sigma''(j)}$ for some permutation $\sigma''\in S_k$. Now the required identity follows by taking $\sigma=\sigma''\sigma'^{-1}$.
\end{proof}

From \Cref{eq:move-permutation}, it follows that the c-cost of a simple move $M$ can be written as
\begin{equation}\label{eq:Kofsimple}
  c(M) = \sum_{i=1}^k c(x_i,y_i) - \sum_{i=1}^k c(x_i,y_{\sigma(i)}) \ .
\end{equation}

The condition in \Cref{eq:Kofsimple} appears in the literature under
the name given in the following definition. This name is due to Rockafellar \cite[\S 24]{rockafellar:1970}, who considered a similar property as a condition for a multi-mapping to be the sub-differential of a convex function.

\begin{definition}
A set of directed edges $G \subset X \times X$ is said to be \emph{cyclically monotone} for the cost $c$ if for each sequence $(x_i,y_j)_{i=1}^k$ in $G$, and each permutation $\sigma \in S_k$, it holds
\begin{equation} \label{eq:monotonicity}
  \sum_{i=1}^k c(x_i,y_i) \leq \sum_{i=1}^k c(x_i,y_{\sigma(i)}) \ .
\end{equation}
\end{definition}

The cyclical monotonicity for the cost $c$ of $\suppof \gamma$ is a
known sufficient and necessary condition for the optimality of $\gamma$
in the corresponding Kantorovich  problem. It is the so-called
Fundamental Theorem of Optimal Transport, see, for example,
\cite[\S~1.6]{santambrogio:2015OTAP}. Here, we want to discuss the
same topic in the algebraic language of moves by using the following
simple equivalence.

\begin{proposition}
A set $G \subset X \times X$ is c-cyclically monotone if, and only if, each simple move $M$ such that $\set{M>0} \in G$ has non-positive value.
\end{proposition}

\begin{proof}
Assume there exists a sequence in $G$ such that \eqref{eq:monotonicity} does not hold. This is equivalent to saying the corresponding move has a positive value and support contained in $G$.
\end{proof}

We restate the Fundamental Theorem as follows. The proof is to be found, for example, in \cite[\S~1.6]{santambrogio:2015OTAP}. We will provide a different proof in the next section.

\begin{proposition}
The coupling $\overline \gamma$ in $\jointof \mu \nu$ has minimal $c$-cost if, and only if, each admissible simple move has a non-positive $c$-cost.
\end{proposition}

Now we briefly discuss the algebraic
properties of simple moves, see \cite{sturmfels:1996}. \Cref{prop:permutations} shows that, given a set $G =
\setof{(x_i,y_i)}{i = 1,\dots, k} \subset X \times X$ and a permutation $\sigma \in S_k$, there exists a simple move $M(G,\sigma) = \sum_{i=1}^k \delta_{x_i}\otimes\delta_{y_i} - \delta_{x_i}\otimes\delta_{y_{\sigma(i)}}$, and, conversely, every simple move is of this type. Notice that the representation is not unique, because if $\sigma(i) = i$, then the two corresponding terms cancel.

Let us consider first the effect of the composition of two permutations. If $\sigma = \pi_1\pi_2$, then
\begin{multline*}
M(G,\sigma) = \sum_{i=1}^k
\delta_{x_i}\otimes\delta_{y_i} -
\delta_{x_i}\otimes\delta_{y_{\sigma(i)}} = \\
\left(\sum_{i=1}^k \delta_{x_i}\otimes\delta_{y_i} -
\delta_{x_i}\otimes\delta_{y_{\pi_2(i)}} \right) +
\left(\sum_{i=1}^k \delta_{x_i}\otimes\delta_{y_{\pi_2(i)}} -
\delta_{x_i}\otimes\delta_{y_{\pi_1\pi_2(i)}}\right) = \\ M(G,\pi_2) +
M(\pi_2 G,\pi_1) \ ,
\end{multline*}
where $\pi_2 G =  \setof{(x_i,y_{\pi_2(i)})}{i=1,\dots,k}$.

Now, every permutation is a product of circular permutations. Consider for example, the case $\sigma = \pi_1\pi_2$, where $\pi_1,\pi_2$ are circular permutations with support $I_1$ and $I_2$, respectively. Choose a coding such that $I_1 = \set{1,\dots,h}$, $I_2=\set{h+1,\dots,k}$. It follows that
\begin{equation*}
M(G,\sigma) = \sum_{i=1}^h \left(\delta_{x_i}\otimes\delta_{y_i} -
\delta_{x_i}\otimes\delta_{y_{i+1}}\right) + \sum_{j=h+1}^k
\left(\delta_{x_j}\otimes\delta_{y_{j}} - \delta_{x_j}\otimes\delta_{y_{j+1}}\right) \ .
\end{equation*}

That is, every simple move is the sum of simple moves associated to a
circular permutation on disjoint supports. In turn, this shows that
the support of a simple move is a union of cycles.

Last case to consider is the case of a permutation given as a product
of exchanges. If $\pi = (i \leftrightarrow j)$, and $G =
\set{(x_1,y_1),(x_2,y_2)}$, then the simple move is $\delta_{x_1}
\otimes \delta_{y_1} + \delta_{x_2}
\otimes \delta_{y_2} - \delta_{x_1}
\otimes \delta_{y_2} - \delta_{x_2}
\otimes \delta_{y_1}$, which is, in fact, a basic move. Indeed, every
simple move is the sum of basic moves. This is a representation
different from that obtained by considering a linear basis because the
representing basic moves depend on the original simple move. They are
not restricted to be elements of a basis.

We conclude this section highlighting that the optimality is related with the
existence of cycles in the support of the coupling, as the following
proposition suggests.

\begin{proposition} \label{cicliottimi}
Let $\gamma \in {\mathcal P}(\mu,\nu)$ be a coupling such that
$\suppof{\gamma}$ contains a cycle and assume that the cost is a
distance, denoted by $d$. Then there exists a coupling $\gamma^*\in
\jointof \mu \nu$ such that $d(\gamma^*)\le d(\gamma)$ and
$\gamma^* - \gamma$ is proportional to a simple move.
\end{proposition}

\begin{proof}
First assume that $\suppof{\gamma}$ has a cycle with two elements of the form  $x_1 \leftrightarrows x_2$.
In this case the basic move $(\delta_{x_1}-\delta_{x_2})\otimes(\delta_{x_2}-\delta_{x_1})$ clearly deletes the cycle and reduces the cost, with $\alpha = \min\{\gamma(x_1,x_2),\gamma(x_2,x_1)\}$.

Assume now that $\suppof{\gamma}$ contains a cycle of length greater than 2. Two cases arise.

If there are two concordant consecutive arrows of the form $x_1
\rightarrow x_2 \rightarrow x_3$, then the move
$\left(\delta_{x_1}  \otimes
\delta_{x_2} + \delta_{x_2}  \otimes \delta_{x_3}\right) -\left( \delta_{x_1}
\otimes \delta_{x_3}+ \delta_{x_2}  \otimes \delta_{x_2}\right) $, with $\alpha = \min\{\gamma(x_1,x_2),\gamma(x_2,x_3)\}$, is
admissible and reduces the cost by virtue of the triangular inequality,
\begin{equation*}
 d(x_1,x_2)+d(x_2,x_3)-d(x_1,x_3)-d(x_2,x_2) \geq 0 \ .
\end{equation*}
Moreover, applying this move, the original cycle is replaced by a
cycle with one edge less.

Finally, if all consecutive edges of $\suppof{\gamma}$ are discordant, such as in
\begin{equation*}
x_1 \rightarrow x_2 \leftarrow x_3 \rightarrow x_4 \leftarrow x_5 \rightarrow x_6 \leftarrow x_1 \, ,
\end{equation*}
then an integer move (not necessarily basic) can be applied both with positive and negative sign. For the example above, the relevant move is
\begin{equation*}
(\delta_{x_1}\otimes \delta_{x_2}+\delta_{x_3}\otimes \delta_{x_4}+\delta_{x_5}\otimes \delta_{x_6}) - (\delta_{x_1}\otimes \delta_{x_6}+\delta_{x_3}\otimes \delta_{x_2}+\delta_{x_5}\otimes \delta_{x_4}) \, .
\end{equation*}
Choosing a sign such that the cost does not increase, and
\begin{equation*}
    \alpha=\min\{\gamma(1,2),\gamma(3,4),\gamma(5,6)\}
    \quad \mbox{or} \quad \alpha=\{\gamma(1,6),\gamma(3,2),\gamma(5,4)\}
\end{equation*}
depending on the sign, one edge of the circuit is deleted.

Notice that all the moves used to reduce a cycle do not produce new cycles because their supports are contained in the relevant cycle.
\end{proof}

\section{Couplings, homophily, and moves} \label{sec:gini}

Early in the $20^{th}$ century, Gini \cite{gini:1914dissomiglianza} defined the notion of index of homophily for a sample $(x_i,y_i)_{i=1}^N$ of a bi-variate real random variable $(X,Y)$. His aim was to discuss a general notion of statistical dependence by comparing the value of $\expectof{\aval{X-Y}}$ with its minimum and maximum value in the class of joint probability functions with the same margins. Based on that, Gini introduced an associated statistical index that was extensively studied in the following years by himself and by others, especially by Salvemini \cite{salvemini:1939} and Dall'Aglio \cite{dallaglio:1956}. A modern account of the Gini methods is to be found in the monograph by Yitzhaki and Schechtman \cite{yitzhaki|schechtman:2013}. Below we describe his work in the context of the subsequent developments by Kantorovich, who was inspired more by early work by Monge on OT than by Gini's methodological ideas. Here we use Gini's method as an intermediate tool to solve more general Kantorovich problems.

Given a bi-variate real sample $(x_i,y_i)_{i=1}^N$, let us sort in ascending order both the first and the second variables, respectively,
\begin{gather*}
x_{(1)} \le x_{(2)} \le \ldots \le x_{(N)} \ , \\
y_{(1)} \le y_{(2)} \le \ldots \le y_{(N)} \ .
\end{gather*}
This operation produces a new bi-variate sample $(x_{(i)},y_{(i)})$, $i=1,\dots,N$, with the same marginal sample distributions as the original one. Gini calls it the \emph{co-graduation} of the original sample.

Clearly, this is a special case of the general theory of coupling, because the original discrete sample distribution and its co-graduation have the same margins.

The difference between the original sample distribution and the co-graduation is the simple move
\begin{equation*}
  \sum_{i=1}^N \delta_{x_i} \otimes \delta_{y_i} - \sum_{i=1}^N \delta_{x_{\sigma'(i)}} \otimes \delta_{y_{\sigma''(i)}} \ ,
\end{equation*}
where $\sigma'$ and $\sigma''$ are permutations of $S_N$ that provide the sorting of each of the two sequences.

More generally, we can say that two finite real sequences $f, g \colon \set{1,\dots,N}  \to \reals$ are \emph{co-monotone} (resp. \emph{counter-monotone}) if
\begin{equation*}
  (f(i) - f(i+1))(g(i)-g(i+1)) \geq  0 \ (\text{resp.} \leq 0) \ , \quad  i = 1, \dots, N-1\ .
\end{equation*}

Clearly, two finite real sequences are co-monotone if they are co-graduated, and two co-monotone sequences are turned into two co-graduated sequences by a suitable common permutation.

We observe that, if a joint probability function has rational probabilities, then it can be simulated by a finite sequence of couplings.
The following proposition is the original Gini's theorem. Notice that the theorem provides a special case of cyclical monotonicity for the distance $d(x,y)=|x-y|$.

\begin{proposition} \label{prop:gini14}
  Given a finite real double sequence $(x_i,y_i)_{i=1}^N$, with joint sample distribution $\gamma$ and marginal distributions $\mu$ and $\nu$, the joint distribution of each bi-variate sequence \begin{equation*}(x_{\sigma'(i)},y_{\sigma'(i)})_{i=1}^N \ , \quad \sigma',
  \sigma'' \in S_N
\end{equation*}
is a coupling of $(\mu,\nu)$. The index
  \begin{equation*}
    M_G(\sigma',\sigma'') = \sum_{i=1}^N \aval{x_{\sigma'(i)} - y_{\sigma''(i)}}
  \end{equation*}
is minimum when the two sequences are co-monotone and is maximum when they are counter-monotone.
\end{proposition}

\begin{proof}
It is enough to consider (as Gini himself does) the co-graduated (respectively counter-graduated) case. Consider each pair of successive indices $i$ and $i+1$. Note first that both
\begin{equation*}
  \aval{x_{\sigma'(i)}-x_{\sigma'(i+1)}}+\aval{y_{\sigma''(i)}-y_{\sigma''(i+1)}} \ \text{and} \ \aval{x_{\sigma'(i)}-y_{\sigma''(i+1)}}+\aval{y_{\sigma''(i)}-x_{\sigma'(i+1)}}
  \end{equation*}
  have the lower bound
  \begin{equation*}\aval{(x_{\sigma'(i)}+y_{\sigma''(i)})-(x_{\sigma'(i+1)}+y_{\sigma''(i+1)})} \ .
\end{equation*}

Enumeration of all possible cases of signs of the differences shows that the minimum is actually the lower bound above and it occurs when the two sequences are co-monotone.
\end{proof}

\begin{remark}
From the point of view of transport theory, we have found that the coupling of maximal index is obtained through the cross-tabulation of the two co-graduated marginal distributions. In modern terms, we can say that Gini has found the $L^1$-optimal coupling of the two marginal distributions when the frequencies are rational.
\end{remark}
\begin{example}\label{ex:homo2-count}
  Assume the bi-variate distribution is represented in a table where the values of the two margins are ordered. If the marginal counts are $4, 6, 2, 4$, for the first variable, and $2, 11, 2, 1$, for the second one, then the co-graduation of the two variables is
\begin{equation*}
  \begin{array}{c|cccccccccccccccc}
t & 1 & 2 & 3 & 4 & 5 & 6 & 7 & 8 & 9 & 10 & 11 & 12 & 13 & 14 & 15 & 16 \\
    \hline
    x & 1 & 1 & 1 & 1 & 2 & 2 & 2 & 2 & 2 & 2 & 3 & 3 & 4 & 4 & 4 & 4 \\
    y & 1 & 1 & 2 & 2 & 2 & 2 & 2 & 2 & 2 & 2 & 2 & 2 & 2 & 3 & 3 & 4
\end{array}
\end{equation*}

The table of maximal homophily $H$ is obtained by pairing these values,
\begin{equation*}
\begin{blockarray}{cccccc}
\begin{block}{c(cccc)c}
 & 2 & 2 & 0 & 0 & 4 \\
& 0 & 6 & 0 & 0 & 6  \\
H=&0 & 2 & 0 & 0 & 2  \\
  &0 & 1 & 2 & 1 &4 \\
  \end{block}
&2 &11 & 2 &1
\end{blockarray}
\end{equation*}
and $M_G=8$.
\end{example}

Proposition \ref{prop:gini14} states that
\begin{equation*}
\sum_{i,j} |a_i-b_j| n(i,j) - \sum_{i,j} |a_i-b_j| n_{\text{co}} (i,j) \geq 0
\end{equation*}
where $a_i$ and $b_j$ are the values of the two margins, respectively, and $n(i,j)$ and $n_{\text{co}} (i,j)$ are the counts in the original table and in $H$, respectively. The previous argument applies to tables of counts, that is, when the frequencies are rational numbers.

More generally, the table $H$ of the example above could be derived from the margins by the so called North-West rule, that is, moving left to right and top to bottom each cell gets the maximum value compatible with the marginal constraints. See the history of the earlier results in \cite{dallaglio:1991advances}. We are going to see that the North-West rule does produce the maximal homophily coupling in the general discrete case.

In the following, without restriction of generality, consider the case where both the values of $x$ and $y$ are $\set{1,\dots,n}$. In this way we have a natural total order on the sample space.

\begin{proposition} \label{pr:homosampling}
Let $H = [n(i,j)]_{i,j =1}^n$ be the maximal homophily table. Then for all pairs $(i,j)$ it holds
\begin{align}
  n(i,j) &= \min \left\{n(i,+) - \sum_{k<j} n(i,k), \ n(+,j)-\sum_{h<i} n(h,j)\right\} \nonumber \\ &=
  \min \left\{\sum_{k \geq j} n(i,k), \ \sum_{h \geq i} n(h,j)\right\} \ .\label{eq:gini-is-nw-2} \end{align}
\end{proposition}

\begin{proof}
  For each pair of indices $(i,j)$, consider $(h,j)$, $h > i$, and $(i,k)$, $k > j$. Let us show that $n(h,j)$ and $n(i,k)$ cannot be both positive. In fact, assume there exists $t_1$ and $t_2$ such that $x_{t_1} = h, y_{t_1} = j, x_{t_2} = i, y_{t_2} = k$. Necessarily, $t_1 \neq t_2$. As $x$ is non-decreasing and $x_{t_1} > x_{t_2}$, it holds $t_1 > t_2$. As $y$ is non-decreasing and $y_{t_1} < y_{t_2}$, it holds $t_1 < t_2$. We have obtained a contradiction and we have shown that only one of the two counts left and down can be positive.

  More precisely, if $n(i,\overline k)>0$ for some $\overline k > j$ then $n(h,j) = 0$ for all $h > i$, that is, if the rest of the row is not all zero, then the rest of the column is. The same holds exchanging rows and columns.

To conclude, write \Cref{eq:gini-is-nw-2} as
\begin{equation*}
n(i,j) = \min \left\{\sum_{k \geq j} n(i,k), \ \sum_{h \geq i} n(h,j)\right\} =n(i,j) + \min \left\{  \sum_{k > j} n(i,k), \ \sum_{h > i} n(h,j)\right\}
\end{equation*}
and observe that at least one among $\sum_{k > j} n(i,k)$ and $\sum_{h > i} n(h,j)$ is zero.
\end{proof}

\begin{proposition} \label{prophomo}
Given two probability functions $\mu$ and $\nu$ on $X$, the lexicographic recursion
\begin{equation} \label{eq:homo}
    \gamma_H(i,j) = \min\left\{\mu(i) - \sum_{k<j} \gamma_H(i,k), \ \nu(j)-\sum_{h<i}\gamma_H(h,j) \right\} \ , \quad i,j \in X \ ,
\end{equation}
uniquely defines the \emph{homophily coupling} $\gamma_H \in \mathcal P(\mu,\nu)$.

\end{proposition}

\begin{proof}
  First note that \Cref{eq:homo} is well defined because the right hand side of the equation involves pairs of indices which precede the current one $(i,j)$.

We want $\gamma_H$ to be non-negative with margins $\mu$ and $\nu$, and $\sum_{i,j} \gamma_H(i,j) = 1$. To prove the proposition, we proceed by recursion on the lines. Consider the first element $\gamma_H(1,1)=\min\{\mu(1),\nu(1)\}$. If $\mu(1)=\nu(1)$, then $\gamma_H(1,1)$ equals the common value and all other elements in the first row and in the first column are zero. Consider now the square sub-table with $i,j = 2,\dots, N$ with the given marginal values. In the case $\mu(1) < \nu(1)$, then $\gamma_H(1,1) = \mu(1)$ and all the other elements of the first row are zero. The sub-table with $i > 1$ has the original first margin and second margin equal to $\nu(1) - \mu(1), \nu(2), \dots, \nu(N)$.  The last case is $\gamma_H(1,1)=\nu(1) \leq \mu(1)$, when all the other entries of the first column are zero. Suppose now that in the first row the entries until the position $\overline k - 1$ are $\gamma_H(1,k)=\nu(k)$ and
$\gamma_H(1,\overline k)=\mu_1-\sum_{k < \overline k} \nu(k)$. The subsequent entries of the first row are zero, and the sum of the first row is equal to $\mu(1)$.

Now consider the sub-table with $n-1$ rows and $(n-\overline k+1)$ columns. The row and column margins of such a table are:
\[
\left(\mu(2), \ldots, \mu(n)\right) \ \textrm{ and } \ \left(-\mu(1) + \sum_{k \le \overline k} \nu(k) , \nu(\overline k + 1), \ldots , \nu(n)\right)
\]
respectively, and the table sums up to $1-\mu(1)$.

As the above procedure does not depend on the normalization of the margins, we can apply the procedure iteratively.
\end{proof}

\begin{example}
Let us consider the probability functions $\mu=(0.5,0.1,0.1,0.3)$ and $\nu=(0.2,0.2,0,2,0.4)$. The $H$-coupling is
\begin{equation*}
\begin{blockarray}{cccccc}
\begin{block}{c(cccc)c}
 & 0.2 & 0.2 & 0.1 & 0&0.5 \\
 & 0 &   0   & 0.1 & 0&0.1\\
\gamma_H(\mu,\nu)=& 0 &   0   &  0  & 0.1 & 0.1 \\
& 0 &  0   &  0  & 0.3 &0.3 \\ \end{block}
&  0.2 & 0.2 & 0.2 & 0.4
\end{blockarray}
\end{equation*}
\end{example}

\begin{theorem} \label{th:connected}
Given two couplings $\gamma, \widetilde{\gamma} \in {\mathcal{P}}(\mu,\nu)$ there exist a sequence of basic moves $M_1, \ldots, M_k$ and a sequence of real positive numbers $\alpha_1, \ldots, \alpha_k$ such that
\begin{equation*}
    \widetilde{\gamma} = \gamma- \sum_{i=1}^k \alpha_i M_i
\end{equation*}
and
\begin{equation*}
    \gamma - \sum_{i=1}^{\overline{k}} \alpha_i M_i \in {\mathcal{P}(\mu,\nu)}
\end{equation*}
for all $\overline{k}=1,\ldots,k$.
\end{theorem}

Noticing that the $H$-coupling is unique in $\mathcal{P}(\mu,\nu)$, the proof of the theorem rests on the following proposition.

\begin{proposition} \label{pr:homo2}
Given a coupling $\gamma \in \mathcal{P}(\mu,\nu)$, there exist a sequence of basic moves $M_1, \ldots, M_k$ and a sequence of real positive numbers $\alpha_1, \ldots, \alpha_k$ such that
\begin{equation*}
    \gamma_H(\mu,\nu) = \gamma - \sum_{i=1}^k \alpha_i M_i
\end{equation*}
and
\begin{equation*}
    \gamma - \sum_{i=1}^{\overline{k}} \alpha_i M_i \in {\mathcal{P}(\mu,\nu)}
\end{equation*}
for all $\overline{k}=1,\ldots,k$.
\end{proposition}

\begin{proof}
We scan the table $\gamma$ from $(1,1)$ to $(1,n)$ in the first row, then from $(2,1)$ to $(2,n)$ in the second row and so on.

Let us consider the probability $\gamma(i,j)$. If
\begin{equation}
    \gamma(i,j) < \min\left\{\mu(i) - \sum_{k<j} \gamma(i,k), \ \nu(j)-\sum_{h<i}\gamma(h,j) \right\}
\end{equation}
then there exist indices $i_1>i$ and $j_1>j$ such that
\begin{equation*}
\gamma(i,j_1)>0 \qquad \gamma(i_1,j)>0 \, .
\end{equation*}
Thus we can apply the basic move $M_{i,i_1,j,j_1}$ with $+1$ in $(i,j_1)$ and $(i_1,j)$, and $-1$ in $(i,j)$ and $(i_1,j_1)$. Let $\alpha_{i,i_1,j,j_1} = \min\{ \gamma_{i,j_1},\gamma_{i_1,j} \}$ and we move from $\gamma$ to $\gamma-\alpha_{i,i_1,j,j_1}M_{i,i_1,j,j_1}$.

Notice that for a given $(i,j)$ only a finite number of moves can be applied since at each step one probability in the $i$-th row or in the $j$-th column goes to zero, and therefore the procedure ends in a finite number of steps.
\end{proof}

In the following remark we show that the Euclidean distance in ${\mathbb R}$ is a typical case where the optimal coupling is not unique.

\begin{remark}
Let us consider the following couplings with $\mu=(0.1,0.25,0.25,0.4)$ and $\nu=(0.5,0.2,0.2,0.1)$.
\begin{equation*}
\gamma_H=
\begin{pmatrix}
  0 & 0 & 0 & 0  \\
0.25 & 0 & 0 & 0   \\
0.25 & 0 & 0 & 0   \\
  0 & 0.25 & 0.25 & 0
\end{pmatrix}\quad
\gamma_1=
\begin{pmatrix}
  0 & 0 & 0 & 0 \\
0 & 0.25& 0 & 0   \\
0.25 & 0 & 0 & 0  \\
0.25 & 0 & 0.25 & 0
\end{pmatrix}
\end{equation*}
\begin{equation*}
\gamma_D=
\begin{pmatrix}
  0 & 0 & 0 & 0 \\
0 & 0.25& 0 & 0  \\
0 & 0 & 0.25 & 0   \\
0.50 & 0 & 0& 0
\end{pmatrix}
 \end{equation*}
We observe that if the ground set is $X=\{1,2,3,4\}$ with the Euclidean distance $d(i,j)=|i-j|$, then all the three couplings have the same c-cost, namely $c(\gamma)=1.5$, which is also equal to the Kantorovich distance. Although this example is rather special, because it has one row and one column with zero probability, nevertheless it allows us to show an example with several couplings sharing the same c-cost by means of small tables.

Notice that the coupling $\gamma_H$ is the coupling of maximum homophily, while the coupling $\gamma_D$ has the highest possible concentration on the main diagonal.

Moreover, all the mixtures of the three previous couplings have again $c(\gamma)=1.5$, showing that the set of the optimal couplings is a face of the polytope. This derives from the fact that with $d(i,j)=|i-j|$ the basic moves involving one diagonal cell, namely of the form $M_{i_1,i_2,i_2,j_2}$, with $i_1<i_2<j_2$, have a null Kantorovich value.
\end{remark}

The following proposition highlights an interesting connection between the discrete and the continuous frameworks for the case of the Euclidean distance. In the discrete case the optimality of the $H$-table follows from previous results, and the optimality in the continuous case is derived.

\begin{proposition}
Given any pair of non-decreasing real sequences $(x_i)_{i=1}^N$, $(y_i)_{i=1}^N$, with sample marginal distributions $\mu_N$ and $\nu_N$, respectively, the homophily coupling $\gamma_H$ coincides with the distribution of $(x_i,y_i)_{i=1}^N$ and hence it minimizes
\begin{equation*}
  \sum_{i,j \in X} \aval{x_i-y_j} \gamma(i,j)
\end{equation*}
among all couplings in $\mathcal P(\mu_N,\nu_N)$. In general, given
any pair of discrete probability functions $\mu$ and $\nu$, $\gamma_H(\mu,\nu)$ is optimal for the Euclidean distance in ${\mathbb R}$.
\end{proposition}

\begin{proof}
The first part follows directly from Proposition
\ref{pr:homosampling}. The second part follows from the continuity of
$(\mu,\nu) \mapsto K_c(\mu,\nu)$, see \Cref{prop:continuity}.
\end{proof}

The following result shows that the directed forest $\suppof{\widetilde\gamma}$ generically contains all loops, that is, vertices for which $\widetilde\gamma(x,x) > 0$.

\begin{proposition}
\label{prop:extreme-nocycles}\
Assume $\mu(\overline x)\nu(\overline x) > 0$ for some $\overline x \in X$. If $\widetilde \gamma$ is an optimal coupling with $\widetilde\gamma(\overline x,\overline x) = 0$, there exists an optimal coupling $\overline\gamma$ with $\overline\gamma(\overline x,\overline x) > 0$ and $\overline\gamma(x,x) = \widetilde\gamma(x,x)$ for $x \neq \overline x$.
\end{proposition}

\begin{proof}
  Assume $\widetilde \gamma$ is optimal and that for a vertex, say 1, it holds $\widetilde \gamma(1,1) = 0$. Since $\mu(1)\nu(1) > 0$, there exist points, say 2 and 3, for which $\widetilde \gamma(1,2) \widetilde \gamma(3,1) > 0$. Pick up the move
  \begin{equation*}
   M =  \delta_1 \otimes \delta_2 + \delta_3 \otimes \delta_1 - \delta_3 \otimes \delta_2 - \delta_1 \otimes \delta_1
  \end{equation*}
  as well as any number $\alpha \in (0, \min\{  \widetilde \gamma(1,2), \widetilde \gamma(3,1)\}]$.

  It is easily checked that the function $\gamma_\alpha = \widetilde \gamma - \alpha M \in \mathcal{P(\mu,\nu)}$ whose value is
  \begin{equation*}
    d(\mu,\nu) - \alpha (d(1,2) + d(3,1) - d(3,2)) \ ,
  \end{equation*}
  and where $d(1,2) + d(3,1) - d(3,2) \geq 0$ is true by the triangle inequality. The equality must hold, otherwise the value would be strictly smaller than the K-distance. In conclusion, $\gamma_\alpha$ is an optimal coupling with $\gamma_\alpha(1,1) > 0$ and with all the other diagonal elements equal to those of the original $\widetilde \gamma$.
\end{proof}

\begin{remark}
By repeating the previous argument, we can show that in the case of $\mu$ and $\nu$ with full support there exists an optimal solution with positive diagonal elements. It should be noticed from the necessary equality $d(1,2) + d(3,1) = d(3,2)$ that solutions with zero elements on the diagonal are not generic.
\end{remark}

\begin{remark}
Notice that the previous proposition is no longer true if we replace a distance with a dissimilarity. Let us consider for example the probability functions $\mu=(0.5,0.2,0.3)$ and $\nu =(0.3,0.2,0.5)$. Moreover, let $X$ be equipped with the following dissimilarity matrix
\begin{equation*}
\begin{pmatrix}
0 & 1 & 5 \\
1 & 0 & 1 \\
5 & 1 & 0
\end{pmatrix}
\end{equation*}
An optimal $\gamma$ is
\begin{equation*}
\gamma_1=\begin{pmatrix}
0.3 & 0.2 & 0 \\
0 & 0 & 0.2 \\
0 & 0 & 0.3
\end{pmatrix}
\end{equation*}
with one null diagonal entry. If we apply a basic move in order to fill in the second diagonal element, we obtain the coupling
\begin{equation*}
\gamma_2=\begin{pmatrix}
0.3 & 0 & 0.2 \\
0 & 0.2 & 0 \\
0 & 0 & 0.3
\end{pmatrix}
\end{equation*}
which is not optimal.
\end{remark}

Next proposition asserts that the support of an optimal coupling
  is generically a connected graph. A detailed study how the support
  of an optimal coupling depends on the given distance has
  been made in \cite{montrucchio|pistone:2019-arXiv:1905.07547v5}.

\begin{proposition}\label{prop:partitions}
If the support of the optimal coupling $\widetilde \gamma$ is a disconnected graph, with connected components $(X_i,\mathcal S_i)$, $i =1,\dots,k$, then $\mu(X_i) = \nu(X_i)$ for all $i = 1,\dots,k$ and $\widetilde \gamma = \sum_{i=1}^k \gamma_i$, where each $\gamma_i$ is supported by $X_i \times X_i$ and is proportional to an optimal coupling for the conditional margins, $\left.\mu\right|_{X_i}$ and $\left.\nu\right|_{X_i}$.
\end{proposition}

\begin{proof} Without restriction of generality, we consider the
    case $k=2$. Assume the supporting graph of $\widetilde \gamma$ has
    components $(X_1,\mathcal S_1)$ and $(X_2,\mathcal S_2)$. This means that $\widetilde\gamma(x,y) = 0$ unless
    $x$ and $y$ belong both to $X_1$ or both to $X_2$. It follows
    that
  \begin{equation*}
 \mu(X_1) = \sum_{x_1 \in X_1}\mu(x_1) = \sum_{x_1, y_1 \in X_1} \widetilde \gamma(x_1,y_1) = \sum_{y_1 \in X_1} \nu(y_1) = \nu(X_1) \ ,
\end{equation*}
and, for the same reason,
$\mu(X_2) = \sum_{x_2, y_2 \in X_2} \widetilde \gamma(x_2,y_2)=
\nu(X_2)$. Now, the K-distance takes the conditional form
\begin{multline*} K_d(\mu,\nu) =
  \sum_{x_1,y_1 \in X_1} d(x_1,y_1) \widetilde \gamma(x_1,y_1) +
  \sum_{x_2,y_2 \in X_2} d(x_2,y_2) \widetilde \gamma(x_2,y_2) = \\
  \sum_{x_1,y_1 \in X_1} \widetilde \gamma(x_1,y_1) \sum_{x_1,y_1 \in
    X_1} d(x_1,y_1) \left.\widetilde \gamma\right|_{X_1\times
    X_1}(x_1,y_1) \ + \\ \sum_{x_2,y_2 \in X_2}  \widetilde \gamma(x_2,y_2) \sum_{x_2,y_2 \in X_2} d(x_2,y_2) \left.\widetilde \gamma\right|_{X_2\times X_2}(x_2,y_2) \ .
\end{multline*}
Each of the conditioned couplings $\left.\widetilde
  \gamma\right|_{X_i\times X_i}$, $i=1,2$ is a coupling of the
conditioned margins $\left.\mu\right|_{X_i}$ and
$\left.\nu\right|_{X_i}$, and such couplings are necessarily optimal.
\end{proof}

\section{Multidimensional extension} \label{sec:three}

In this section we extend the results in Proposition \ref{prophomo} to the case of joint probability functions with three given margins.

\begin{proposition} \label{prophomo3}
Given three probability functions $\mu$, $\nu$, and $\zeta$, the joint probability function $\gamma_H$ such that
\begin{multline}
\label{eq:homo3}
    \gamma_H(i,j,k) = \min\left\{\mu(i) - \sum_{(\overline j,\overline k) \prec (j,k)} \gamma_H(i,\overline j, \overline k), \right. \\ \left. \nu(j)- \sum_{(\overline i,\overline k) \prec (i,k)} \gamma_H(\overline i, j, \overline k) , \ \zeta(k) - \sum_{(\overline i,\overline j) \prec (i,j)} \gamma_H(\overline i,\overline j, k)\right\}
\end{multline}
is well defined, and it is unique. We name this joint probability function as the joint probability function of maximal homophily. In \Cref{eq:homo3} the sign $\prec$ is to be read in lexicographic order, e.g., $(\overline j,\overline k)\prec (j,k)$ if and only if either $\overline j<j$ or $\overline j=j \mbox{ and } \overline k<k$.
\end{proposition}

\begin{proof}
We prove that the definition in Equation \eqref{eq:homo3} is the lifting of the bi-variate $H$-coupling for $\mu$ and $\nu$ with respect to $\zeta$. In other words, the definition for three variables is iterative. Consider the coupling $\gamma_H(i,j)$ with the indices $(i,j)$ in lexicographic order, and build the table of maximal homophily for $\gamma_H(i,j)$ and $\zeta$. From \Cref{eq:homo}, we have
\begin{equation*}
\gamma_H(i,j,k) = \min\left\{ \zeta(k) - \sum_{(\overline i,\overline j ) \prec (i,j)} \gamma_H(\overline i, \overline j, k), \gamma_H(i,j) - \sum_{\overline k < k} \gamma_H(i,j,\overline k) \right\}.
\end{equation*}
We check that the expression above leads to \Cref{eq:homo3}.
\begin{multline*}
\gamma_H(i,j,k) = \min\left\{ \zeta(k) - \sum_{(\overline i,\overline j ) \prec (i,j)} \gamma_H(\overline i, \overline j, k),\right. \\ \left. \min\left\{\mu(i)-\sum_{\overline j<j} \gamma_H(i,\overline j), \nu(j) - \sum_{\overline i<i} \gamma_H(\overline i,j)\right\}- \sum_{\overline k < k} \gamma_H(i,j,\overline k)  \right\}
\end{multline*}
\begin{multline*}
= \min\left\{ \zeta(k) - \sum_{(\overline i,\overline j ) \prec (i,j)} \gamma_H(\overline i, \overline j, k), \right. \\ \left. \mu(i)-\sum_{\overline j<j} \gamma_H(i,\overline j) - \sum_{\overline k < k} \gamma_H(i,j,\overline k), \nu(j) - \sum_{\overline i<i} \gamma_H(\overline i,j)- \sum_{\overline k < k} \gamma_H(i,j,\overline k) \right\}
\end{multline*}
\begin{multline*}
= \min\left\{ \zeta(k) - \sum_{(\overline i,\overline j ) \prec (i,j)} \gamma_H(\overline i, \overline j, k),\right. \\ \left.\mu(i) - \sum_{(\overline j,\overline k) \prec (j,k)} \gamma_H(i,\overline j, \overline k),   \nu(j)- \sum_{(\overline i,\overline k) \prec (i,k)} \gamma_H(\overline i, j, \overline k) \right\}
\end{multline*}

\end{proof}

\begin{example}
Consider a joint sample distribution of three variables with mar\-ginal counts $4, 6, 2, 4$ for the first variable and $2, 11, 2, 1$ for the second variable, as  in Example \ref{ex:homo2-count}, and  $3,3,5,5$ for the third one.
The co-graduation of the three variables is
\begin{equation*}
  \begin{array}{c|cccccccccccccccc}
t & 1 & 2 & 3 & 4 & 5 & 6 & 7 & 8 & 9 & 10 & 11 & 12 & 13 & 14 & 15 & 16 \\
    \hline
    x & 1 & 1 & 1 & 1 & 2 & 2 & 2 & 2 & 2 & 2 & 3 & 3 & 4 & 4 & 4 & 4 \\
    y & 1 & 1 & 2 & 2 & 2 & 2 & 2 & 2 & 2 & 2 & 2 & 2 & 2 & 3 & 3 & 4 \\
    z & 1 & 1 & 1 & 2 & 2 & 2 & 3 & 3 & 3 & 3 & 3 & 4 & 4 & 4 & 4 & 4
\end{array}
\end{equation*}
The four slices of the table of maximal homophily are:
\begin{center}
\begin{tabular}{cccc}
$z=1$&$z=2$&$z=3$&$z=4$\\
$\begin{pmatrix}
  2 & 1 & 0 & 0 \\
  0 & 0 & 0 & 0 \\
  0 & 0 & 0 & 0 \\
  0 & 0 & 0 & 0
\end{pmatrix}$
&
$\begin{pmatrix}
  0 & 1 & 0 & 0 \\
  0 & 2 & 0 & 0 \\
  0 & 0 & 0 & 0 \\
  0 & 0 & 0 & 0
\end{pmatrix}$
&
$ \begin{pmatrix}
  0 & 0 & 0 & 0 \\
  0 & 4 & 0 & 0 \\
  0 & 1 & 0 & 0 \\
  0 & 0 & 0 & 0
\end{pmatrix}$
&
$\begin{pmatrix}
  0 & 0 & 0 & 0 \\
  0 & 0 & 0 & 0 \\
  0 & 1 & 0 & 0 \\
  0 & 1 & 2 & 1
\end{pmatrix}
$
\end{tabular}
\end{center}
\end{example}

\medskip

We now introduce the basic moves in the tri-variate case and we prove that they are enough to connect all joint probability functions, using the same arguments as in the bi-variate case. To ease the notation, we write only the indices and we omit the symbol $\delta$ when considering the moves.

There are two types of basic moves: in the first type the $+1$ have a common index, while in the second type the $+1$ have all different indices.

\begin{definition}
Consider indices $i \ne i'$, $j \ne j'$, $k \ne k'$. The tri-variate basic moves on $X \times X \times X$ are of two types:
\begin{itemize}
    \item[T1:] $+1$ in $(i,j,k)$ and in $(i,j',k')$, $-1$ in $(i,j',k)$ and in $(i,j,k')$, and similarly the second $+1$ in $(i',j,k')$ or in $(i',j',k)$;

    \item[T2:] $+1$ in $(i,j,k)$ and in $(i',j',k')$ and $-1$
    \begin{itemize}
        \item in $(i,j',k')$ and in $(i',j,k)$ or
        \item in $(i',j,k')$ and in $(i,j',k)$ or
        \item in $(i',j',k)$ and in $(i,j,k')$.
    \end{itemize}
\end{itemize}
\end{definition}

Two examples of basic moves are pictured in Figure \ref{fig:bm3}.

\begin{figure}
\centering
\includegraphics[width=.65\textwidth]{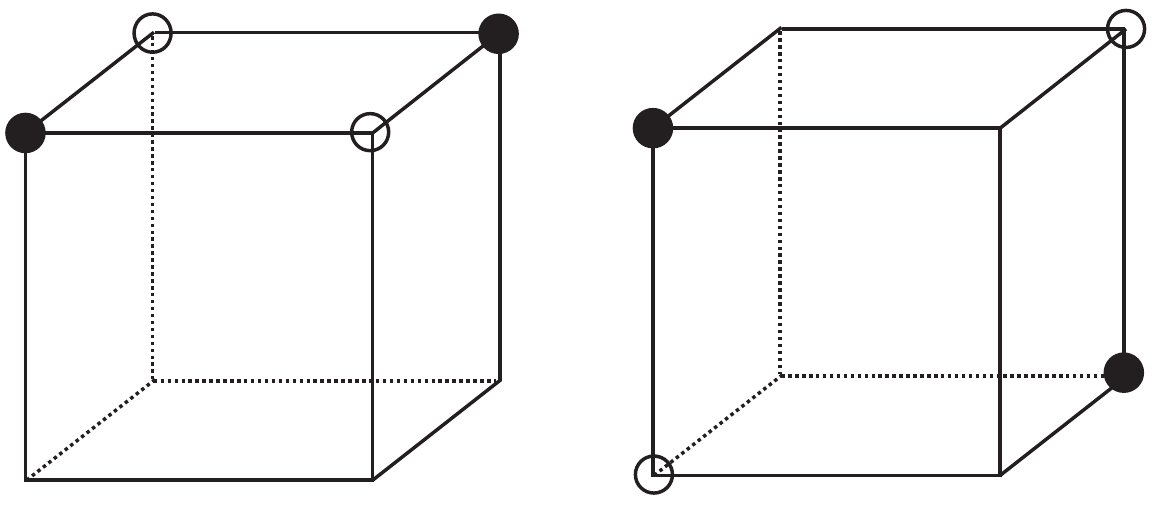}
\caption{Two basic moves in the tri-variate case. On the left, a move of the type T1, on the right a move of the type T2. Black circles correspond to $+1$, white circles to $-1$.\label{fig:bm3}}
\end{figure}

We are now ready to extend Theorem \ref{th:connected} to the tri-variate case.

\begin{theorem}
Given two tri-variate joint probability functions $\gamma, \widetilde{\gamma} \in {\mathcal{P}}(\mu,\nu,\zeta)$ there exist a sequence of tri-variate basic moves $M_1, \ldots, M_k$ and a sequence of real positive numbers $\alpha_1, \ldots, \alpha_k$ such that
\begin{equation*}
    \widetilde{\gamma} = \gamma - \sum_{i=1}^k \alpha_i M_i
\end{equation*}
and
\begin{equation*}
    \gamma - \sum_{i=1}^{\overline{k}} \alpha_i M_i \in {\mathcal{P}(\mu,\nu,\zeta)}
\end{equation*}
for all $\overline{k}=1,\ldots,k$.
\end{theorem}

\begin{proof}
We prove that from each joint probability function we can reach the maximal homophily by using basic moves, following the same strategy as in the proof of Theorem \ref{pr:homo2}.

If the condition in \Cref{eq:homo3} is not satisfied, then there is an entry $(i,j,k)$ such that
\begin{equation*}
    \gamma(i,j',k')>0, \ \gamma(i'',j,k'')>0, \ \gamma(i''',j''',k)>0
\end{equation*}
for suitable indices $i',i'' \ne i$; $j'',j''' \ne j$; $k',k''' \ne k$.

Let $\alpha= \min\{\gamma(i,j',k'), \gamma(i'',j,k''), \gamma(i''',j''',k) \}$.

Then, define the integer move $M$ with
\begin{itemize}
    \item $-1$ in $(i,j,k)$, $(i'',j',k')$ and $(i''',j''',k'')$;
    \item $+1$ in $(i,j',k')$, $(i'',j,k'')$ and $(i''',j''',k)$.
\end{itemize}
Such a move, applied with the coefficient $\alpha$ above, satisfies the condition in \Cref{eq:homo3} in the point $(i,j,k)$. The new points in $(i'',j',k')$ and $(i''',j''',k'')$ are lexicographically greater than $(i,j,k)$, so that scanning the joint probability function from $(1,1,1)$ lexicographically the procedure ends in a finite number of steps.

Finally, note that if the move $M$ lies in a slice (i.e., $i=i'=i''$ or $j=j'=j''$ or $k=k'=k''$) the move $M$ is a basic move since one $+1$ and one $-1$ coincide. In the other cases, the move $M$ can be decomposed into two basic moves:
\begin{itemize}
    \item $M_1$ with $-1$ in $(i,j,k)$ and $(i'',j',k')$,  $+1$ in $(i,j',k')$ and $(i'',j,k)$;
    \item $M_2$ with $-1$ in $(i'',j,k)$ and $(i''',j''',k'')$, $+1$ in $(i'',j,k'')$ and $(i''',j''',k)$.
\end{itemize}
\end{proof}

\section{Algorithm} \label{sec:algo}

The Simulated Annealing for continuous variables has been introduced in \cite{vanderbilt:84}, then optimized in several ways for special applications. In its basics, a Simulated Annealing algorithm seeks to find the minimum of a real function through a Markov chain whose stationary distribution is uniform on the set of the global minima. At each step, the Markov chain moves in a suitable set of neighbours and the transition probability is selected in order to have the desired stationary distribution. For further details, see \cite{henderson:2003}.

The basic moves introduced in the previous sections allow us to define the neighbours and to obtain a connected chain. Moreover, we exploit the special properties of the Kantorovich function, and through Proposition \ref{prop:extreme-nocycles} we perform one further optimization step.

The pseudo-code of the algorithm is given in Figure \ref{fig:alg}. To simplify the presentation, we write the algorithm in the case of two-dimensional joint probability functions, but it can be easily adapted to the three-dimensional case.

\begin{figure}
\begin{algorithm}[H]
\DontPrintSemicolon

  \KwInput{Two vectors $\mu$ and $\nu$}
  \KwOutput{Configuration with minimum Kantorovich value}
  \KwData{Kantorovich function $K$, initial temperature $\tau_0$, number of MCMC steps $B$}
  Initialize: $\gamma=\mu \otimes \nu$; $\tau=\tau_0$\\
  \For{$b$ in $1:B$}
    {
        Choose an admissible basic move $M$ with $-1$ in $(i_1,j_1)$ and $(i_2,j_2)$\\
        Compute $\alpha=\min\{\gamma(i_1,j_2),\gamma(j_2,i_1)\}$\\
        Generate $u$ uniform in $[0,\alpha]$\\
        Define $\gamma'=\gamma-u\cdot M$ \\
        \If{$\gamma' \ge 0$}
        {Define $p=\min\left\{\exp(-(K(\gamma')-K(\gamma)))^{1/\tau}, 1 \right\}$ \\
         Generate $v$ uniform in $[0,1]$ \\
         \If{$p>v$}
         {
         $\gamma=\gamma'$
         }
        }
        Decrease $\tau$
    }
    \ForEach{$M$ with $i_1<i_2=j_1<j_2$ or $j_1<j_2=i_1<i_2$}
    {
     \If{$\alpha=\min\{\gamma(i_1,j_2)\gamma(i_2,j_1)\}>0$}
     {$\gamma=\gamma-\alpha \cdot M$}
    }
\caption{Simulated Annealing with final optimization step}
\end{algorithm}
\caption{Pseudo-code of the algorithm.} \label{fig:alg}
\end{figure}

To choose the simulation parameters (i.e., the initial temperature $\tau_0$ and the length of the Markov chain $B$), we have performed a preliminary simulation study for values of $n$ ranging from $4$ to $20$. In the set $\{1, \ldots, n\}$, the distance used here is $d(i,j)=\sqrt{|i-j|}$.

In the first part of the simulation study, we have computed the acceptance probability of the first move of the MCMC as a function of the initial temperature $\tau_0$. The results are displayed in Table \ref{tab:ap}. Each value is based on a sample of $10,000$ pairs of marginal probability functions $\mu,\nu$. Each entry of $\mu,\nu$ is chosen under the uniform distribution $\mathcal U [0,1]$, and the two vectors are  then normalized.

\begin{remark}
Our Simulated Annealing implementation has the independence coupling as its starting point. This is because it is a joint probability distribution far from the vertices of the polytope.
\end{remark}

The initial temperature $\tau_0$ can be chosen reasonably small. For instance, if we fix $0.95$ as the acceptance probability of the first move, $\tau_0$ decreases with $n$ and ranges from $10^{-0.6}$ for $n=4$ to $10^{-2.0}$ for $n=20$.

In the second part of the simulation study, we have inspected when the Markov chain does not produce new moves to evaluate the convergence of the algorithm. For values of the number $B$ of the MCMC steps ranging from $10$ to $10^5$, we have computed how many moves would be accepted in a window of $100$ further steps. The simulation is based on $1,000$ pairs of marginal probability functions $\mu, \nu$ in each case, randomly chosen as in the previous part of the study.  The initial temperature for each $n$ has been chosen from the first part of the study, as outlined above. The temperature decrease function used here is $\tau=\tau_0(0.95)^b$, $b=1,\ldots,B$, but similar results are obtained for other choices, namely $\tau=\tau_0(0.99)^b$, $\tau=\tau_0/b$, $\tau=\tau_0/\log(1+b)$.

The proportions of accepted moves  are displayed in Table \ref{tab:end}. We observe that for values of the number $B$ of the MCMC steps ranging from $10^3$ and $10^5$ the acceptance probability of a new move is less than $0.001$.

\begin{table}
\small{
\begin{tabular}{|c|*{6}r r|}
  \hline
& $10^{-2.6}$ & $10^{-2.4}$ & $10^{-2.2}$ & $10^{-2.0}$ & $10^{-1.8}$ & $10^{-1.6}$ & $10^{-1.4}$ \\ \hline
$4$ & $0.5998$ & $0.6406$ & $0.6665$ & $0.7150$ & $0.7587$ & $0.8049$ & $0.8503$ \\
$5$ & $0.6515$ & $0.6850$ & $0.7257$ & $0.7770$ & $0.8138$ & $0.8567$ & $0.8944$  \\
$6$ & $0.6919$ & $0.7236$ & $0.7698$ & $0.8098$ & $0.8585$ & $0.8929$ & $0.9191$ \\
$7$ & $0.7260$ & $0.7621$ & $0.8036$ & $0.8451$ & $0.8822$ & $0.9139$ & $0.9405$  \\
$8$ & $0.7485$ & $0.7879$ & $0.8270$ & $0.8664$ & $0.9021$ & $0.9302$ & $0.9520$ \\
$9$ & $0.7724$ & $0.8103$ & $0.8532$ & $0.8880$ & $0.9142$ & $0.9420$ & $0.9572$  \\
$10$ & $0.7943$ & $0.8357$ & $0.8671$ & $0.9005$ & $0.9263$ & $0.9496$ & $0.9650$  \\
$11$ & $0.8142$ & $0.8516$ & $0.8815$ & $0.9113$ & $0.9371$ & $0.9571$ & $0.9702$  \\
$12$ & $0.8186$ & $0.8648$ & $0.8947$ & $0.9217$ & $0.9430$ & $0.9629$ & $0.9742$  \\
$13$ & $0.8352$ & $0.8702$ & $0.9064$ & $0.9287$ & $0.9520$ & $0.9672$ & $0.9764$  \\
$14$ & $0.8530$ & $0.8847$ & $0.9133$ & $0.9368$ & $0.9564$ & $0.9702$ & $0.9799$ \\
$15$ & $0.8631$ & $0.8931$ & $0.9206$ & $0.9392$ & $0.9606$ & $0.9734$ & $0.9818$  \\
$16$ & $0.8636$ & $0.8991$ & $0.9267$ & $0.9466$ & $0.9627$ & $0.9759$ & $0.9847$  \\
$17$ & $0.8757$ & $0.9082$ & $0.9301$ & $0.9513$ & $0.9674$ & $0.9776$ & $0.9850$  \\
$18$ & $0.8834$ & $0.9126$ & $0.9337$ & $0.9533$ & $0.9702$ & $0.9790$ & $0.9869$ \\
$19$ & $0.8927$ & $0.9173$ & $0.9419$ & $0.9587$ & $0.9710$ & $0.9808$ & $0.9878$  \\
$20$ & $0.8937$ & $0.9195$ & $0.9446$ & $0.9629$ & $0.9746$ & $0.9827$ & $0.9886$  \\ \hline
 \hline
\end{tabular}
}
\medskip

\small{
\begin{tabular}{|c|*{6}r r|}
  \hline
& $10^{-1.2}$ & $10^{-1.0}$ & $10^{-0.8}$ & $10^{-0.6}$ & $10^{-0.4}$ & $10^{-0.2}$ & $10^{0}$ \\ \hline
$4$ & $0.8901$ & $0.9217$ & $0.9437$ & $0.9639$ & $0.9766$ & $0.9844$ & $0.9900$ \\
$5$ & $0.9250$ & $0.9473$ & $0.9671$ & $0.9763$ & $0.9842$ & $0.9903$ & $0.9936$ \\
$6$ & $0.9450$ & $0.9639$ & $0.9751$ & $0.9836$ & $0.9895$ & $0.9934$ & $0.9958$ \\
$7$ & $0.9583$ & $0.9721$ & $0.9817$ & $0.9878$ & $0.9922$ & $0.9950$ & $0.9968$ \\
$8$ & $0.9660$ & $0.9780$ & $0.9856$ & $0.9905$ & $0.9937$ & $0.9963$ & $0.9976$ \\
$9$ & $0.9729$ & $0.9824$ & $0.9883$ & $0.9926$ & $0.9953$ & $0.9970$ & $0.9980$  \\
$10$ & $0.9765$ & $0.9856$ & $0.9899$ & $0.9940$ & $0.9961$ & $0.9975$ & $0.9984$ \\
$11$ & $0.9804$ & $0.9880$ & $0.9920$ & $0.9948$ & $0.9967$ & $0.9979$ & $0.9986$ \\
$12$ & $0.9833$ & $0.9894$ & $0.9932$ & $0.9955$ & $0.9971$ & $0.9982$ & $0.9989$ \\
$13$ & $0.9859$ & $0.9907$ & $0.9943$ & $0.9960$ & $0.9975$ & $0.9984$ & $0.9991$ \\
$14$ & $0.9868$ & $0.9915$ & $0.9945$ & $0.9966$ & $0.9978$ & $0.9986$ & $0.9991$ \\
$15$ & $0.9884$ & $0.9925$ & $0.9950$ & $0.9969$ & $0.9981$ & $0.9988$ & $0.9992$ \\
$16$ & $0.9901$ & $0.9931$ & $0.9961$ & $0.9974$ & $0.9983$ & $0.9989$ & $0.9993$ \\
$17$ & $0.9903$ & $0.9938$ & $0.9963$ & $0.9975$ & $0.9985$ & $0.9991$ & $0.9994$ \\
$18$ & $0.9917$ & $0.9944$ & $0.9964$ & $0.9977$ & $0.9986$ & $0.9991$ & $0.9994$ \\
$19$ & $0.9925$ & $0.9949$ & $0.9968$ & $0.9981$ & $0.9988$ & $0.9992$ & $0.9995$ \\
$20$ & $0.9930$ & $0.9956$ & $0.9971$ & $0.9981$ & $0.9989$ & $0.9992$ & $0.9995$ \\ \hline
 \hline
\end{tabular}
}

\caption{Acceptance probability of the first move of the MCMC for $4 \leq n \leq 20$ and initial temperature $10^{-2.6}\leq \tau_0 \leq 1$.} \label{tab:ap}
\end{table}

\begin{table}
\small{
\begin{tabular}{|c|*{4}r r|}
  \hline
& $10^1$ & $10^2$ & $10^3$ & $10^4$ & $10^5$ \\ \hline
4 & 0.0532 & 0.0079 & 0.0000 & 0.0000 & 0.0000\\
5 & 0.0829 & 0.0235 & 0.0004 & 0.0000 & 0.0000\\
6 & 0.1122 & 0.0382 & 0.0020 & 0.0000 & 0.0000\\
7 & 0.1437 & 0.0553 & 0.0049 & 0.0000 & 0.0000\\
8 & 0.1698 & 0.0730 & 0.0092 & 0.0000 & 0.0000\\
9 & 0.1975 & 0.0869 & 0.0129 & 0.0002 & 0.0000\\
10 & 0.2232 & 0.1054 & 0.0182 & 0.0006 & 0.0000\\
11 & 0.2470 & 0.1235 & 0.0213 & 0.0011 & 0.0000\\
12 & 0.2686 & 0.1403 & 0.0279 & 0.0020 & 0.0000\\
13 & 0.2867 & 0.1577 & 0.0325 & 0.0024 & 0.0000\\
14 & 0.3038 & 0.1754 & 0.0370 & 0.0036 & 0.0000\\
15 & 0.3194 & 0.1947 & 0.0417 & 0.0044 & 0.0000\\
16 & 0.3358 & 0.2055 & 0.0477 & 0.0051 & 0.0001\\
17 & 0.3440 & 0.2242 & 0.0524 & 0.0063 & 0.0001\\
18 & 0.3602 & 0.2406 & 0.0585 & 0.0072 & 0.0002\\
19 & 0.3690 & 0.2481 & 0.0623 & 0.0090 & 0.0004\\
20 & 0.3790 & 0.2664 & 0.0681 & 0.0097 & 0.0007\\
 \hline
\end{tabular}
}
\caption{Proportion of accepted moves after $B$ iterations of the MCMC for $4 \leq n \leq 20$ and $B$ from $10^1$ to $10^5$.} \label{tab:end}
\end{table}

\begin{example}\label{ex:opt_coup}
For $n=10$ consider the two margins $\mu$ and $\nu$:
\begin{eqnarray*}
& (0.0732,0.0976,0.1220,0.1463,0.1707,0.0244,0.0488,0.0732,0.0976,0.1463) , \\
& (0.2059,0.0000,0.0294,0.0882,0.1471,0.1176,0.0588,0.1765,0.0882,0.0882) .
\end{eqnarray*}
Using the distance $d(i,j)=\sqrt{|i-j|}$, and the numerical parameters from the simulation study, the algorithm produces in less than one second the optimal coupling shown in Table \ref{tab:ex_opt}, whose c-cost is $0.4648$.

\begin{table}
\begin{footnotesize}
\begin{equation*}
\begin{pmatrix}
0.0732 & 0.0000 & 0.0000 & 0.0000 & 0.0000 & 0.0000 & 0.0000 & 0.0000 & 0.0000 & 0.0000\\
0.0976 & 0.0000 & 0.0000 & 0.0000 & 0.0000 & 0.0000 & 0.0000 & 0.0000 & 0.0000 & 0.0000\\
0.0352 & 0.0000 & 0.0294 & 0.0000 & 0.0000 & 0.0115 & 0.0100 & 0.0359 & 0.0000 & 0.0000\\
0.0000 & 0.0000 & 0.0000 & 0.0882 & 0.0000 & 0.0581 & 0.0000 & 0.0000 & 0.0000 & 0.0000\\
0.0000 & 0.0000 & 0.0000 & 0.0000 & 0.1471 & 0.0237 & 0.0000 & 0.0000 & 0.0000 & 0.0000\\
0.0000 & 0.0000 & 0.0000 & 0.0000 & 0.0000 & 0.0244 & 0.0000 & 0.0000 & 0.0000 & 0.0000\\
0.0000 & 0.0000 & 0.0000 & 0.0000 & 0.0000 & 0.0000 & 0.0488 & 0.0000 & 0.0000 & 0.0000\\
0.0000 & 0.0000 & 0.0000 & 0.0000 & 0.0000 & 0.0000 & 0.0000 & 0.0732 & 0.0000 & 0.0000\\
0.0000 & 0.0000 & 0.0000 & 0.0000 & 0.0000 & 0.0000 & 0.0000 & 0.0093 & 0.0882 & 0.0000\\
0.0000 & 0.0000 & 0.0000 & 0.0000 & 0.0000 & 0.0000 & 0.0000 & 0.0581 & 0.0000 & 0.0882
\end{pmatrix}
\end{equation*}
\end{footnotesize}
\caption{Optimal coupling of Example \ref{ex:opt_coup} found by the Simulated Annealing.} \label{tab:ex_opt}
\end{table}
\end{example}

\section*{Acknowledgments}
The Authors thank Luigi Montrucchio (Collegio Carlo Alberto) for useful bibliographical suggestions and insightful comments on an early version of this paper.
The Authors are grateful to the referee for his/her helpful suggestions. G. Pistone gratefully acknowledges the support of de Castro Statistics and Collegio Carlo Alberto.


\begin{thebibliography}{10}

\bibitem{aliprantis|border:2006}
Charalambos~D. Aliprantis and Kim~C. Border.
\newblock {\em Infinite dimensional analysis}.
\newblock Springer, Berlin, third edition, 2006.
\newblock A hitchhiker's guide.

\bibitem{aoki|hara|takemura:2012}
Satoshi Aoki, Hisayuki Hara, and Akimichi Takemura.
\newblock {\em Markov bases in algebraic statistics}.
\newblock Springer Series in Statistics. Springer, New York, 2012.

\bibitem{barvinok:2002}
Alexander Barvinok.
\newblock {\em A course in convexity}, volume~54 of {\em Graduate Studies in
  Mathematics}.
\newblock American Mathematical Society, Providence, RI, 2002.

\bibitem{bollobas:1998}
B\'ela Bollob\'as.
\newblock {\em Modern graph theory}, volume 184 of {\em Graduate Texts in
  Mathematics}.
\newblock Springer-Verlag, 1998.

\bibitem{brualdi:2006}
Richard~A. Brualdi.
\newblock {\em Combinatorial matrix classes}, volume 108 of {\em Encyclopedia
  of Mathematics and its Applications}.
\newblock Cambridge University Press, Cambridge, 2006.

\bibitem{dallaglio:1956}
Giorgio Dall'Aglio.
\newblock Sugli estremi dei momenti delle funzioni di ripartizione doppia.
\newblock {\em Ann. Scuola Norm. Sup. Pisa Cl. Sci. (3)}, 10:35--74, 1956.

\bibitem{dallaglio:1991advances}
Giorgio Dall'Aglio.
\newblock Fr\'{e}chet classes: the beginnings.
\newblock In {\em Advances in probability distributions with given marginals
  ({R}ome, 1990)}, volume~67 of {\em Math. Appl.}, pages 1--12. Kluwer Acad.
  Publ., Dordrecht, 1991.

\bibitem{diaconis|sturmfels:98}
Persi Diaconis and Bernd Sturmfels.
\newblock Algebraic algorithms for sampling from conditional distributions.
\newblock {\em The Annals of Statistics}, 26(1):363--397, 1998.

\bibitem{fienberg:80}
Stephen~E. Fienberg.
\newblock {\em The analysis of cross-classified categorical data}.
\newblock MIT Press, second edition, 1980.

\bibitem{gini:1914dissomiglianza}
Corrado Gini.
\newblock Di una misura della dissomiglianza di due gruppi di quantit\`a e
  delle sue applicazioni allo studio delle relazioni statistiche.
\newblock {\em Atti R. Ist. Veneto Sc. Lett. Arti}, LXXIV:185--213, 1914.

\bibitem{henderson:2003}
Darrall Henderson, Sheldon~H. Jacobson, and Alan~W. Johnson.
\newblock The theory and practice of simulated annealing.
\newblock In Fred Glover and Gary~A. Kochenberger, editors, {\em Handbook of
  Metaheuristics}, pages 287--319. Springer, Boston, MA, 2003.

\bibitem{montrucchio|pistone:2019-arXiv:1905.07547v5}
Luigi Montrucchio and Giovanni Pistone.
\newblock Kantorovich distance on a weighted graph.
\newblock arXiv:1905.07547 [math.PR], 2019.

\bibitem{peyre|cuturi:2019}
Gabriel Peyr\'e and Marco Cuturi.
\newblock Computational optimal transport.
\newblock {\em Foundations and Trends in Machine Learning}, 11(5--6):355--607,
  2019.
\newblock arXiv:1803.00567v2.

\bibitem{pistone|riccomagno|wynn:2001}
Giovanni Pistone, Eva Riccomagno, and Henry~P. Wynn.
\newblock {\em Algebraic statistics: Computational commutative algebra in
  statistics}, volume~89 of {\em Monographs on Statistics and Applied
  Probability}.
\newblock Chapman \& Hall/CRC, Boca Raton, FL, 2001.

\bibitem{rapallo:2003-SJS}
Fabio Rapallo.
\newblock Algebraic {M}arkov bases and {MCMC} for two-way contingency tables.
\newblock {\em Scand. J. Statist.}, 30(2):385--397, 2003.

\bibitem{rockafellar:1970}
R.~Tyrrell Rockafellar.
\newblock {\em Convex analysis}.
\newblock Princeton Mathematical Series, No. 28. Princeton University Press,
  1970.

\bibitem{salvemini:1939}
Tommaso Salvemini.
\newblock Sugli indici di omofilia.
\newblock {\em Supplemento Statistico ai nuovi problemi}, 5:105--115, 1939.

\bibitem{santambrogio:2015OTAP}
Filippo Santambrogio.
\newblock {\em Optimal {T}ransport for {A}pplied {M}athematicians: {C}alculus
  of {V}ariations, {PDE}s, and {M}odeling}.
\newblock Birkh\"{a}user, 2015.

\bibitem{sturmfels:1996}
Bernd Sturmfels.
\newblock {\em Gr{\"o}bner bases and convex polytopes}.
\newblock American Mathematical Society, 1996.

\bibitem{sullivan:2018book}
Seth Sullivant.
\newblock {\em Algebraic Statistics}.
\newblock Number 194 in Graduate Studies in Mathematics. AMS, 2018.

\bibitem{vanderbilt:84}
David Vanderbilt and Steven~G. Louie.
\newblock A {M}onte carlo simulated annealing approach to optimization over
  continuous variables.
\newblock {\em Journal of Computational Physics}, 56(2):259--271, 1984.

\bibitem{villani:2003-topics}
C\'{e}dric Villani.
\newblock {\em Topics in optimal transportation}, volume~58 of {\em Graduate
  Studies in Mathematics}.
\newblock American Mathematical Society, Providence, RI, 2003.

\bibitem{yitzhaki|schechtman:2013}
Shlomo Yitzhaki and Edna Schechtman.
\newblock {\em The {G}ini methodology}.
\newblock Springer Series in Statistics. Springer, New York, 2013.
\newblock A primer on a statistical methodology.

\end{thebibliography}

\end{document}